\documentclass[11pt]{article}

\usepackage{amsthm}
\usepackage[letterpaper, margin=1in]{geometry}
\usepackage{times}
\usepackage[backend=biber, style=alphabetic, citestyle=alphabetic, maxalphanames=6, maxcitenames=99, mincitenames=98, maxbibnames=99, giveninits=true]{biblatex}
\usepackage[margin=10pt, font=small, labelfont=bf]{caption}

\usepackage{amsmath, amssymb}

\usepackage{xpatch}

\makeatletter
\patchcmd\blx@bblinput{\blx@blxinit}
                      {\blx@blxinit
                      }{}{\fail}
\makeatother

\usepackage{lmodern}

\usepackage{bbm}

\usepackage{url}
\usepackage[table]{xcolor}
\usepackage{array}

\usepackage{setspace}

\usepackage{wrapfig}

\makeatletter
\newcommand{\labeltarget}[1]{\Hy@raisedlink{\hypertarget{#1}{}}}
\makeatother

\usepackage{upgreek}

\usepackage{complexity}

\usepackage[inline]{enumitem}
\usepackage{moreenum}

\usepackage{mathtools}
\usepackage[super]{nth}

\usepackage{bm}

\usepackage{xspace}

\usepackage{ifthen}

\usepackage[immediate]{silence}
\usepackage{sectsty}
\DeactivateWarningFilters[tmp]

\allsectionsfont{\boldmath}

\setlist[enumerate]{nosep, topsep=0.1em}
\setlist[enumerate, 1]{label=(\roman*), leftmargin=2.2em}
\setlist[itemize]{nosep, topsep=0.3em}

\usepackage[textsize=footnotesize, color=blue!30!white]{todonotes}

\usepackage{xfrac}

\usepackage{tikz}
\usetikzlibrary{calc}
\usetikzlibrary{shapes.geometric}
\usetikzlibrary{arrows}
\usetikzlibrary{decorations.pathreplacing, decorations.pathmorphing}

\usepackage{tcolorbox}
\tcbuselibrary{skins}

\usepackage{float}

\usepackage{subcaption}
\DeactivateWarningFilters[caption]

\usepackage{graphicx}
\graphicspath{{graphics/}}
\makeatletter
\newcommand\appendtographicspath[1]{%
  \g@addto@macro\Ginput@path{#1}%
}
\makeatother

\usepackage[vlined, ruled]{algorithm2e}
\DontPrintSemicolon
\SetAlgoSkip{bigskip}
\SetAlgoInsideSkip{smallskip}
\SetKwInput{KwInput}{Input}
\SetKwInput{KwOutput}{Output}

\usepackage{mdframed}

\usepackage[pdfpagelabels,bookmarks=false]{hyperref}
\hypersetup{colorlinks, linkcolor=darkblue, citecolor=darkgreen, urlcolor=darkblue}

\usepackage{aliascnt}
\usepackage[capitalize, nameinlink]{cleveref}

\definecolor{darkblue}{rgb}{0,0,0.38}
\definecolor{darkred}{rgb}{0.8,0,0}
\definecolor{darkgreen}{rgb}{0.1,0.35,0}

\newtheorem{theorem}{Theorem}[section]

\newaliascnt{conjecture}{theorem}
\newtheorem{conjecture}[theorem]{Conjecture}
\aliascntresetthe{conjecture}

\newaliascnt{corollary}{theorem}

\aliascntresetthe{corollary}

\newaliascnt{definition}{theorem}
\newtheorem{definition}[theorem]{Definition}
\aliascntresetthe{definition}

\newaliascnt{lemma}{theorem}
\newtheorem{lemma}[lemma]{Lemma}
\aliascntresetthe{lemma}

\newaliascnt{remark}{theorem}

\aliascntresetthe{remark}

\newaliascnt{observation}{theorem}

\aliascntresetthe{observation}

\newaliascnt{claim}{theorem}
\newtheorem{claim}[theorem]{Claim}
\aliascntresetthe{claim}

\Crefname{theorem}{Theorem}{Theorems}

\Crefname{conjecture}{Conjecture}{Conjectures}
\Crefname{corollary}{Corollary}{Corollaries}
\Crefname{definition}{Definition}{Definitions}
\Crefname{lemma}{Lemma}{Lemmas}
\Crefname{observation}{Observation}{Observations}
\Crefname{remark}{Remark}{Remarks}

\Crefname{line}{Line}{Lines}
\crefalias{AlgoLine}{line}

\DeclareMathOperator{\supp}{supp}
\DeclareMathOperator{\argmax}{argmax}

\newcommand{\BinomDistr}{B}
\DeclareMathOperator{\UnifDistr}{\mathrm{Unif}}
\renewcommand{\E}{\mathbb{E}}
\renewcommand{\P}{\Pr}  %

\newcommand{\Z}{\mathbb{Z}}
\renewcommand{\R}{\mathbb{R}}

\renewcommand{\M}{\mathcal{M}}
\newcommand{\I}{\mathcal{I}}
\newcommand{\mspan}{\mathrm{span}}
\DeclareMathOperator{\rank}{rank}
\newcommand{\mpartition}{\mathcal{P}}
\newcommand\OPT{\ensuremath{\mathrm{OPT}}}
\newcommand{\card}[1]{\lvert#1\rvert}
\newcommand{\mrestrict}[2]{\left.#1\right|_{#2}}
\newcommand{\mminor}[3]{\left.\left(#1 / #2\right)\right|_{#3}}

\newcommand{\rhodownshift}{\rho'}
\newcommand{\rhoapprox}{\tilde{\rho}}
\newcommand{\rhogrid}{\overline{\rho}}

\newcommand{\lambdagrid}{\overline{\lambda}}
\newcommand{\Lambdagrid}{\overline{\Lambda}}
\newcommand{\orig}{\ensuremath\mathrm{orig}}

\renewcommand{\epsilon}{\varepsilon}

\makeatletter
\let\@@pmod\pmod
\DeclareRobustCommand{\pmod}{\@ifstar\@pmods\@@pmod}
\def\@pmods#1{\mkern8mu({\operator@font mod}\mkern 6mu#1)}
\makeatother

\makeatletter
\let\@@mod\mod
\DeclareRobustCommand{\mod}{\@ifstar\@mods\@@mod}
\def\@mods#1{\mkern8mu{\operator@font mod}\mkern 6mu#1}
\makeatother

\allowdisplaybreaks
 
\makeatletter
\def\@fnsymbol#1{\ensuremath{\ifcase#1\or *\or %
\ddagger\or
    \mathsection\or \mathparagraph\or \|\or **\or \dagger\dagger
    \or \ddagger\ddagger \else\@ctrerr\fi}}
\makeatother

\title{Constant-Competitiveness for Random Assignment Matroid Secretary Without Knowing the Matroid%
\thanks{This project received funding from Swiss National Science Foundation grant 200021\_184622 and the European Research Council (ERC) under the European Union's Horizon 2020 research and innovation programme (grant agreement No 817750).}
}

\author{
    Richard Santiago%
    \thanks{%
        Department of Mathematics, ETH Zurich, Zurich, Switzerland.
        Email: \href{mailto:rtorres@ethz.ch}%
        {rtorres@ethz.ch}.%
    }%
    \and
    Ivan Sergeev%
    \thanks{%
        Department of Mathematics, ETH Zurich, Zurich, Switzerland.
        Email: \href{mailto:isergeev@ethz.ch}%
        {isergeev@ethz.ch}.%
    }%
    \and
    Rico Zenklusen
    \thanks{%
        Department of Mathematics, ETH Zurich, Zurich, Switzerland.
        Email: \href{mailto:ricoz@ethz.ch}%
        {ricoz@ethz.ch}.%
    }%
}

\date{}

\begin{document}

\maketitle

\begin{abstract}
The Matroid Secretary Conjecture is a notorious open problem in online optimization. It claims the existence of an $O(1)$-competitive algorithm for the Matroid Secretary Problem (MSP). Here, the elements of a weighted matroid appear one-by-one, revealing their weight at appearance, and the task is to select elements online with the goal to get an independent set of largest possible weight. $O(1)$-competitive MSP algorithms have so far only been obtained for restricted matroid classes and for MSP variations, including \emph{Random-Assignment} MSP (RA-MSP), where an adversary fixes a number of weights equal to the ground set size of the matroid, which then get assigned randomly to the elements of the ground set. Unfortunately, these approaches heavily rely on knowing the full matroid upfront. This is an arguably undesirable requirement, and there are good reasons to believe that an approach towards resolving the MSP Conjecture should not rely on it. Thus, both Soto [SIAM Journal on Computing 2013] and Oveis Gharan \& Vondr{\'a}k [Algorithmica 2013] raised as an open question whether RA-MSP admits an $O(1)$-competitive algorithm even without knowing the matroid upfront.

In this work, we answer this question affirmatively. Our result makes RA-MSP the first well-known MSP variant with an $O(1)$-competitive algorithm that does not need to know the underlying matroid upfront and without any restriction on the underlying matroid. Our approach is based on first approximately learning the rank-density curve of the matroid, which we then exploit algorithmically.
\end{abstract}
 
\begin{tikzpicture}[overlay, remember picture, shift = {(current page.south east)}]
\begin{scope}[shift={(-1.1,1.5)}]
\def\hd{2.5}
\node at (-2.15*\hd,0) {\includegraphics[height=0.7cm]{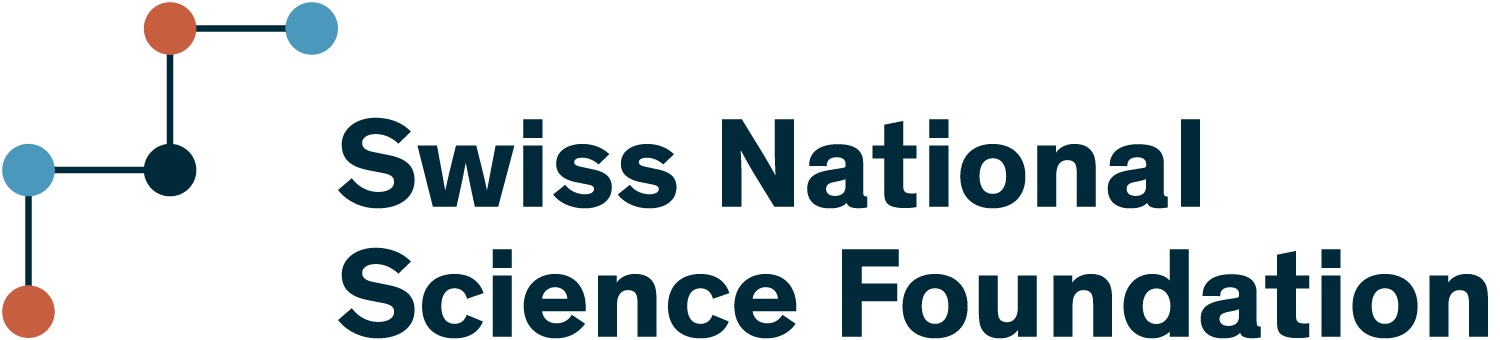}};
\node at (-\hd,0) {\includegraphics[height=1.0cm]{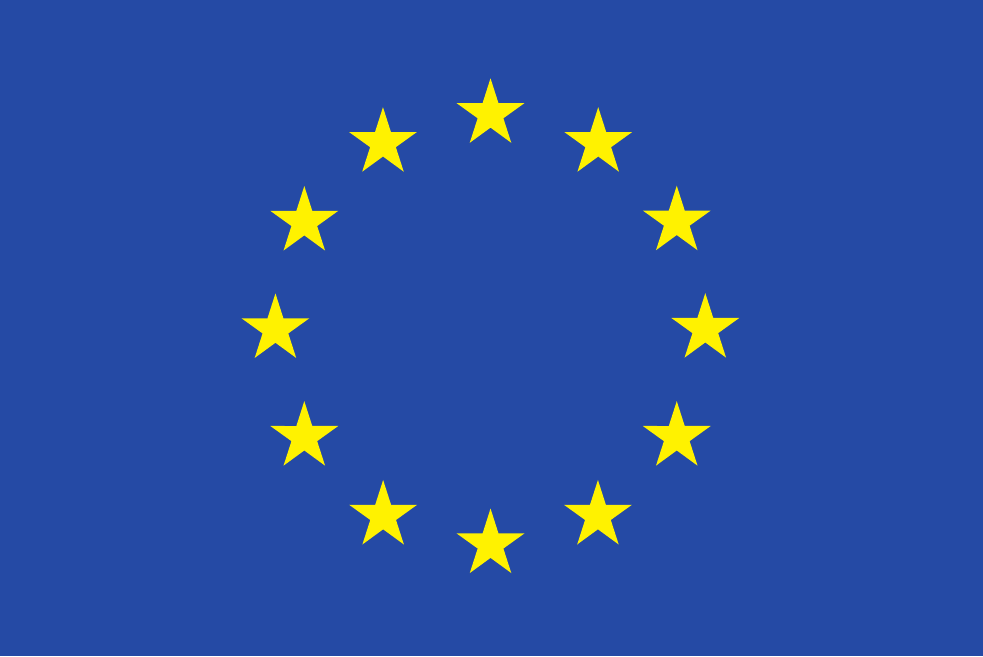}};
\node at (-0.2*\hd,0) {\includegraphics[height=1.2cm]{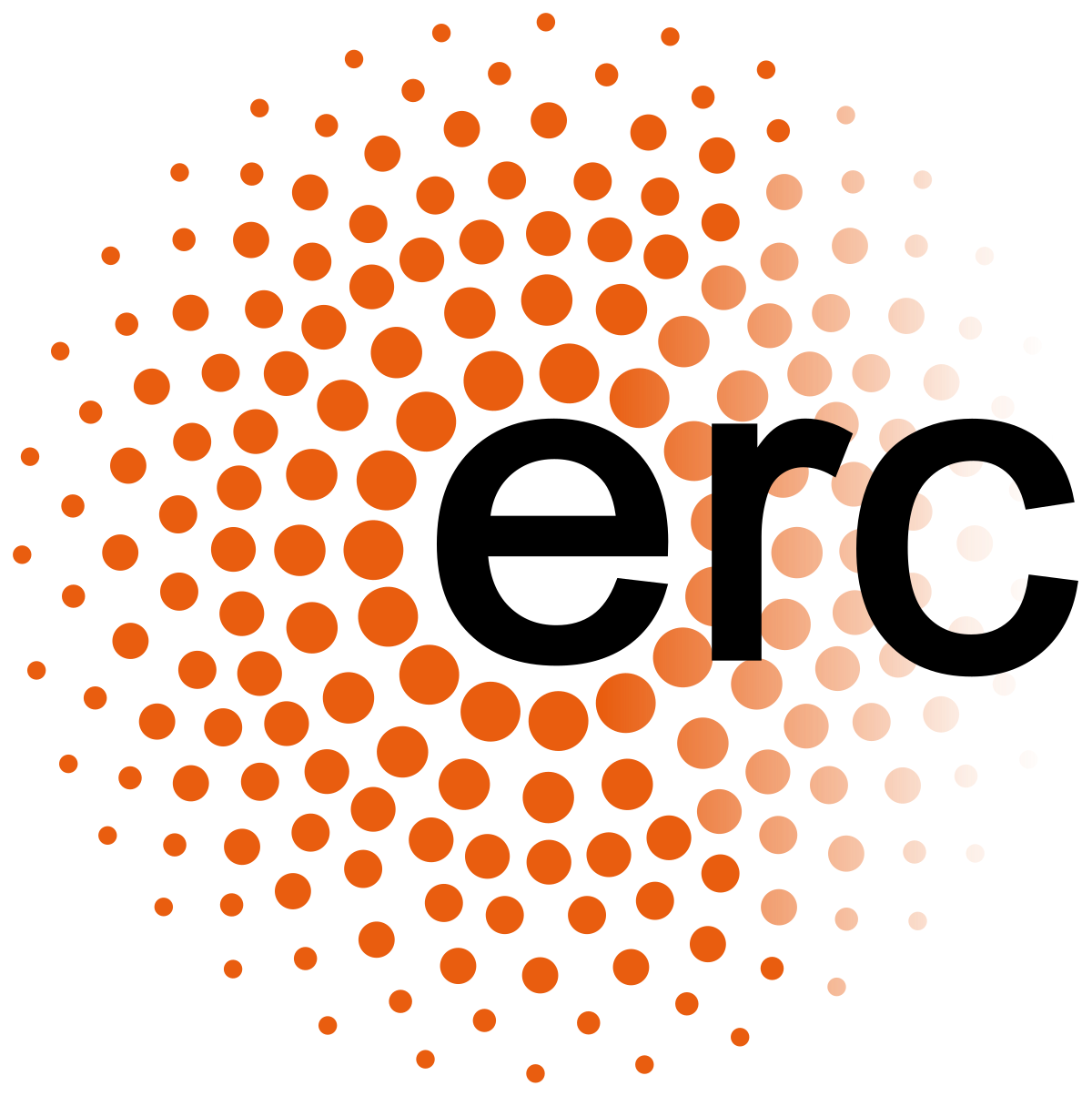}};
\end{scope}
\end{tikzpicture}

\section{Introduction}\label{sec:intro}

The Matroid Secretary Problem (MSP), introduced by~\textcite{babaioff_2007_matroids}, is a natural and well-known generalization of the classical Secretary Problem~\cite{dynkin_1963_secretary}, motivated by strong connections and applications in mechanism design. Formally, MSP is an online selection problem where we are given a matroid $\M = (N, \I)$,%
\footnote{A matroid $\M$ is a pair $\M = (N, \I)$ where $N$ is a finite set and $\I \subseteq 2^{N}$ is a nonempty family satisfying: 1) if $A \subseteq B$ and $B \in \I$ then $A \in \I$, and 2) if $A, B \in \I$ and $\card{B} > \card{A}$ then $\exists e \in B \setminus A$ such that $A \cup \{e\} \in \I$.}
with elements of unknown weights $w \colon N \to \R_{\geq 0}$ that appear one-by-one in uniformly random order. Whenever an element appears, it reveals its weight and one has to immediately and irrevocably decide whether to select it.
The goal is to select a set of elements $I \subseteq N$ that
\begin{enumerate*}
    \item is independent, i.e., $I \in \I$, and
    \item has weight $w(I) = \sum_{e \in I} w(e)$ as large as possible.
\end{enumerate*}
The key challenge in the area is to settle the notorious Matroid Secretary Problem (MSP) Conjecture:

\begin{conjecture}[\cite{babaioff_2007_matroids}]
    There is an $O(1)$-competitive algorithm for MSP.
\end{conjecture}

The best-known procedures for MSP are $O(\log\log(\rank(\M)))$-competitive \cite{lachish_2014_loglog,feldman_2018_simple}, where $\rank(\M)$ is the rank of the matroid $\M$, i.e., the cardinality of a largest independent set.

Whereas the MSP Conjecture remains open, extensive work in the field has led to constant-competitive algorithms for variants of the problem and restricted settings. This includes constant-competitive algorithms for specific classes of matroids~\cite{
    babaioff_2018_matroid,
    korula_2009_algorithms,
    im_2011_secretary,
    soto_2013_matroid,
    jaillet_2013_advances,
    ma_2013_simulatedSTACS,
    dimitrov_2012_competitive,
    kesselheim_2013_optimal,
    dinitz_2013_matroid,
}.
Moreover, in terms of natural variations of the problem, \textcite{soto_2013_matroid} showed that constant-competitiveness is achievable in the so-called \emph{Random-Assignment MSP}, or RA-MSP for short. Here, an adversary chooses $\card{N}$ weights, which are then assigned uniformly at random to the ground set elements $N$ of the matroid. (Soto's result was later extended by~\textcite{oveis_2013_variants} to the setting where the arrival order of the elements is adversarial instead of uniformly random.) Constant-competitive algorithms also exist for the \emph{Free Order Model}, where the algorithm can choose the order in which elements appear~\cite{jaillet_2013_advances}.

Intriguingly, a key aspect of prior advances on constant-competitive algorithms for special cases and variants of MSP is that they heavily rely on knowing the full matroid $\M$ upfront. This is also crucially exploited in \citeauthor{soto_2013_matroid}'s work on RA-MSP. In fact, if the matroid is not known upfront in full, there is no natural variant of MSP for which a constant-competitive algorithm is known.

A high reliance on knowing the matroid $\M = (N, \I)$ upfront (except for its size $\card{N}$) is undesirable when trying to approach the MSP Conjecture, because it is easy to obstruct an MSP instance by adding zero-weight elements. Not surprisingly, all prior advances on the general MSP conjecture, like the above-mentioned $O(\log\log(\rank(\M)))$-competitive algorithms~\cite{lachish_2014_loglog, feldman_2018_simple} and also earlier procedures~\cite{babaioff_2007_matroids, chakraborty_2012_improved}, only need to know $\card{N}$ upfront and make calls to an independence oracle on elements revealed so far. Thus, for RA-MSP, it was raised as an open question, both in~\cite{soto_2013_matroid} and~\cite{oveis_2013_variants}, whether a constant-competitive algorithm exists without knowing the matroid upfront. The key contribution of this work is to affirmatively answer this question, making the random assignment setting the first MSP variant for which a constant-competitive algorithm is known without knowing the matroid and without any restriction on the underlying matroid.

\begin{theorem}\label{thm:main}
    There is a constant-competitive algorithm for RA-MSP with only the cardinality of the matroid known upfront.
\end{theorem}

Moreover, our result holds in the more general \emph{adversarial order with a sample} setting, where we are allowed to sample a random constant fraction of the elements and all remaining (non-sampled) elements arrive in adversarial order.

As mentioned, when the matroid is fully known upfront, an $O(1)$-competitive algorithm was known for RA-MSP even when the arrival order of all elements is adversarial~\cite{oveis_2013_variants}. Interestingly, for this setting it is known that, without knowing the matroid upfront, no constant-competitive algorithm exists. More precisely, a lower bound on the competitiveness of $\Omega(\sfrac{\card{N}}{\log\log\card{N}})$ was shown in~\cite{oveis_2013_variants}.

\paragraph{Organization of the Paper}

We start in \cref{sec:densities} with a brief discussion on the role of (matroid) densities in the context of random assignment models, as our algorithm heavily relies on densities. Decomposing the matroid into parts of different densities has been central in prior advances on RA-MSP. However, this crucially relies on knowing the matroid upfront. We work with a rank-density curve, introduced in \cref{sec:rankdensity}, which is also unknown upfront; however, we show that it can be learned approximately (in a well-defined sense) by observing a random constant fraction of the elements. \cref{sec:outline} provides an outline of our approach based on rank-density curves and presents the main ingredients allowing us to derive \cref{thm:main}. 
\cref{sec:curveProperties} takes a closer look at rank-density curves and shows some of their useful properties. 
\cref{sec:curveEstimate} showcases the main technical tool that allows us to approximate the rank-density curve from a sample set. 
Finally, \cref{sec:algorithm} combines the ingredients to present our final algorithm and its analysis.

We emphasize that we predominantly focus on providing a simple algorithm and analysis, refraining from optimizing the competitive ratio of our procedure at the cost of complicating the presentation.

We assume that all matroids are loopless, i.e., every element is independent by itself.
This is without loss of generality, as loops can simply be ignored in matroid secretary problems.
\section{Random-Assignment MSP and Densities}\label{sec:densities}

A main challenge in the design and analysis of MSP algorithms is how to protect heavier elements (or elements of an offline optimum) from being spanned by lighter ones that are selected earlier during the execution of the algorithm. In the random assignment setting, however, weights are assigned to elements uniformly at random, which allows for shifting the focus from protecting elements based on their weights to protecting elements based on their role in the matroid structure. Intuitively speaking, an element arriving in the future is at a higher risk of being spanned by the algorithm's prior selection if it belongs to an area of the matroid of large cardinality and small rank (a ``dense'' area) than an area of small cardinality and large rank (a ``sparse'' area).

This is formally captured by the notion of density: the \emph{density} of a set $U \subseteq N$ in a matroid $\M = (N, \I)$ is $\sfrac{\card{U}}{r(U)}$, where $r \colon 2^{N} \to \Z_{\geq 0}$ is the rank function of $\M$.%
\footnote{The rank function $r \colon 2^{N} \to \Z_{\geq 0}$ assigns to any set $U \subseteq N$ the cardinality of a maximum cardinality independent set in $U$, i.e., $r(U) \coloneqq \max\{\card{I} \colon I \subseteq U, I \in \I\}$.}

Densities play a crucial role in RA-MSP~\cite{soto_2013_matroid,oveis_2013_variants}. Indeed, prior approaches decomposed $\M$ into its \emph{principal sequence}, which is the chain $\emptyset \subsetneq S_{1} \subsetneq \ldots \subsetneq S_{k} = N$ of sets of decreasing densities obtained as follows. $S_{1} \subseteq N$ is the densest set of $\M$ (in case of ties it is the unique maximal densest set), $S_{2}$ is the union of $S_{1}$ and the densest set in the matroid obtained from $\M$ after contracting $S_{1}$, and so on until a set $S_{k}$ is obtained with $S_{k} = N$. \cref{subfig:principalDecomp} shows an example of the principal sequence of a graphic matroid.

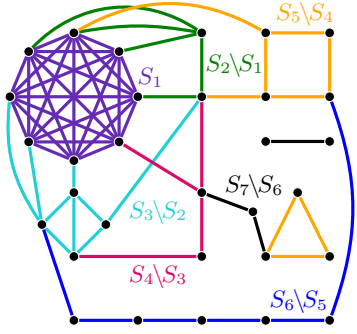
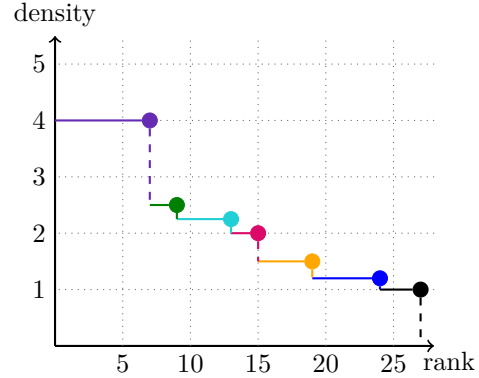
\begin{figure}[ht]
    \centering
    \begin{subfigure}[b]{0.485\linewidth}
        \centering
\definecolor{colora}{RGB}{101,  44, 179}
\colorlet{colorb}{green!50!black}
\definecolor{colorc}{RGB}{ 33, 207, 214}
\definecolor{colord}{RGB}{219,   9, 104}
\colorlet{colore}{orange!70!yellow}
\colorlet{colorf}{blue}
\colorlet{colorg}{black}

\begin{tikzpicture}[
    scale=0.85,
    layer/.style={circle, draw=#1, fill=#1, inner sep=1pt},
    layer/.default={black},
    layera/.style={layer, colora},
    layerb/.style={layer, colorb},
    layerc/.style={layer, colorc},
    layerd/.style={layer, colord},
    layere/.style={layer, colore},
    layerf/.style={layer, colorf},
    layerg/.style={layer, colorg},
]

    \node at (0,-4.2) {};

    \begin{scope}[every node/.style={circle, inner sep=1.3pt, fill=black, draw=white}]
        \foreach \i [evaluate=\i as \a using int(\i * 45)] in {0, ..., 7}{\node (va\i) at (\a:  1) {};}
        
        \node (vb1) at (2, 1) {};
        \node (vb2) at (2, 0) {};
        
        \node (vc1) at (-0.5, -2) {};
        \node (vc2) at (0, -1.5) {};
        \node (vc3) at (0, -2.5) {};
        \node (vc4) at (0.5, -2) {};
        
        \node (vd1) at (2, -1.5) {};
        \node (vd2) at (2, -2.5) {};
        
        \node (ve1) at (3, 1) {};
        \node (ve2) at (4, 1) {};
        \node (ve3) at (4, 0) {};
        \node (ve4) at (3, 0) {};
        \node (ve5) at (3, -2.5) {};
        \node (ve6) at (3.5, -1.5) {};
        \node (ve7) at (4, -2.5) {};
        
        \foreach \x in {0, ..., 4}{\node (vf\x) at (\x, -3.5) {};}
        
        \node (vg1) at (3, -0.7) {};
        \node (vg2) at (4, -0.7) {};
        \node (vg3) at (2.8,-1.8) {};

    \end{scope}

    \begin{scope}[line width=1.2pt]
        \foreach \i in {0, ..., 6}{
            \pgfmathtruncatemacro{\s}{\i + 1}
            \foreach \j in {\s, ..., 7}{
                \draw[layera] (va\i) -- (va\j);
            }
        }

        \draw[layerb] (vb1) edge[bend right=0] (va1)
                      (vb1) edge[bend right=10] (va2)
                      (vb1) edge[bend right=40] (va3)
                      (vb1) edge[bend right=0] (vb2)
                      (vb2) edge[bend right=0] (va0);

        \draw[layerc] (vc1) edge[bend left=20] (va4)
                      (vc1) edge[bend right=0] (va5)
                      (vc1) edge (vc2)
                      (vc1) edge (vc3)
                      (vc2) edge (vc4)
                      (vc2) edge (vc3)
                      (vc3) edge (vc4)
                      (vc2) edge[bend right=0] (va6)
                      (vc4) edge[bend right=0] (vb2);

        \draw[layerd] (vd1) edge (vb2)
                      (vd1) edge (va7)
                      (vd1) edge (vd2)
                      (vd2) edge (vc3);

        \draw[layere] (ve1) edge[bend right=30] (va2)
                      (ve4) edge (vb2)
                      (ve1) edge (ve2)
                      (ve2) edge (ve3)
                      (ve3) edge (ve4)
                      (ve4) edge (ve1)
                      (ve5) edge (ve6)
                      (ve6) edge (ve7)
                      (ve7) edge (ve5);

        \foreach \x [evaluate=\x as \sx using int(\x+1)] in {0, ..., 3}{\draw[layerf] (vf\x) edge (vf\sx);}
        \draw[layerf] (vf0) edge (vc1)
                      (vf4) edge[bend right=20] (ve3);

        \draw[layer] (vg1) edge (vg2);
        \draw[layer] (vd1) edge (vg3);
        \draw[layer] (vg3) edge (ve5);
    \end{scope}

    \begin{scope}[font=\footnotesize]
        \node[colora] at ($(va0)+(0.2,0.3)$) {$S_1$};
        \node[colorb] at ($(vb2)+(-0.1,0.5)$)[right] {$S_2\!\!\setminus\!\! S_1$};
        \node[colorc] at ($(vc4)+(0.2,0.2)$)[right] {$S_3\!\!\setminus\!\! S_2$};
        \node[colord] at ($(vd2)+(-1.3,-0.3)$)[right] {$S_4\!\!\setminus\!\! S_3$};
        \node[colore] at ($(ve1)+(0.0,0.3)$)[right] {$S_5\!\!\setminus\!\! S_4$};
        \node[colorf] at ($(vf3)+(-0.1,0.3)$)[right] {$S_6\!\!\setminus\!\! S_5$};
        \node[colorg] at ($(vg3)+(-0.6,0.4)$)[right] {$S_7\!\!\setminus\!\! S_6$};
    \end{scope}
\end{tikzpicture}
         \caption{Principal decomposition.}\label{subfig:principalDecomp}
    \end{subfigure}%
    \begin{subfigure}[b]{0.515\linewidth}
        \centering
\definecolor{colora}{RGB}{101,  44, 179}
\colorlet{colorb}{green!50!black}
\definecolor{colorc}{RGB}{ 33, 207, 214}
\definecolor{colord}{RGB}{219,   9, 104}
\colorlet{colore}{orange!70!yellow}
\colorlet{colorf}{blue}
\definecolor{colorg}{RGB}{144, 144, 144}

\begin{tikzpicture}[
    xscale=0.18,
    yscale=0.75,
    layer/.style={circle, draw=#1, fill=#1, inner sep=2pt},
    layer/.default={black},
    layera/.style={layer, colora},
    layerb/.style={layer, colorb},
    layerc/.style={layer, colorc},
    layerd/.style={layer, colord},
    layere/.style={layer, colore},
    layerf/.style={layer, colorf},
    layerg/.style={layer, colorg},
    gridline/.style={dotted, color = black!60},
    plotmain/.style={thick, fill=none},
    plothint/.style={thick, dashed, fill=none}
]
\small
    \pgfmathsetmacro\xmin{0.0}
    \pgfmathsetmacro\xmax{28}
    \pgfmathsetmacro\xticks{25}
    \pgfmathsetmacro\ymin{0.0}
    \pgfmathsetmacro\ymax{5.5}
    \pgfmathsetmacro\yticks{5}

    \draw[thick, ->](\xmin, \ymin) -- (\xmax, \ymin) node[pos=1,inner sep=0pt,shift={(2mm,-2mm)}]{rank};
    \draw[thick, ->](\xmin, \ymin) -- (\xmin, \ymax) node[above, pos=1]{density};

    \foreach \i in {5, 10, ..., \xticks}{\draw[gridline] (\i, \ymin) -- (\i, \ymax) node[black, below, pos=0]{\i};}
    \foreach \i in {1, ..., \yticks}{\draw[gridline] (\xmax, \i) -- (\xmin, \i) node[black, left,  pos=1]{\i};}

    \begin{scope}{every node/.style+={fill}}
        \coordinate (v0) at (0, 4);
        \node[layera] (v1) at (7, 4) {};
        \node[layerb] (v2) at (9, 2.5) {};
        \node[layerc] (v3) at (13, 2.25) {};
        \node[layerd] (v4) at (15, 2) {};
        \node[layere] (v5) at (19, 1.5) {};
        \node[layerf] (v6) at (24, 1.2) {};
        \node[layer]  (v7) at (27, 1) {};
        \coordinate (v8) at (27, 0);
    \end{scope}

    \foreach[count=\i] \style [evaluate=\i as \j using int(\i - 1)] in {layera, layerb, layerc, layerd, layere, layerf, black}{
        \draw[plotmain, \style] (v\i -| v\j) -- (v\i);
    }
    \foreach[count=\i] \style [evaluate=\i as \j using int(\i + 1)] in {layera, layerb, layerc, layerd, layere, layerf, black}{
        \draw[plothint, \style] (v\i) -- (v\j -| v\i);
    }
\end{tikzpicture}
         \caption{Rank-density curve.}\label{subfig:rankDensityCurve}
    \end{subfigure}
    \caption{\cref{subfig:principalDecomp} shows a graph representing a graphic matroid together with its principal sequence $\emptyset \subsetneq S_{1} \subsetneq \dots \subsetneq S_{7} = N$, where $N$ are all edges of the graph. \cref{subfig:rankDensityCurve} shows its rank-density curve. Each step in the rank-density curve (highlighted by a circle) corresponds to one $S_{i}$ and has $y$-coordinate equal to the density of $\M_{i} = \mminor{\M}{S_{i - 1}}{S_{i} \setminus S_{i - 1}}$ and $x$-coordinate equal to $r(S_{i})$.}
	\label{fig:rdc-simple}
\end{figure}

Previous approaches then considered, independently for each $i \in [k] \coloneqq \{1, \dots, k\}$, the matroid $\M_{i} \coloneqq \mminor{\M}{S_{i - 1}}{S_i\setminus S_{i - 1}}$, i.e., the matroid obtained from $\M$ by first contracting $S_{i - 1}$ and then restricting to $S_{i} \setminus S_{i - 1}$. (By convention, we set $S_{0} \coloneqq \emptyset$.) These matroids are also known as the \emph{principal minors} of $\M$. Given an independent set in each principal minor, their union is guaranteed to be independent in the original matroid $\M$. Prior approaches (see, in particular,~\cite{soto_2013_matroid} for details) then exploited the following two key properties of the principal minors $\M_{i}$:

\begin{enumerate}
    \item\label{item:decompOffOptOk} $\sum_{i = 1}^{k} \E[w(\OPT(\M_{i}))] = \Omega(\E[w(\OPT(\M))])$, where $\OPT(\M)$ (and analogously $\OPT(\M_{i})$) is an (offline) maximum weight independent set in $\M$ and the expectation is over all random weight assignments.
    \item\label{item:decompGivesUnifDense} Each matroid $\M_{i}$ is \emph{uniformly dense}, which means that the (unique maximal) densest set in $\M_{i}$ is the whole ground set of $\M_{i}$.
\end{enumerate}

Property~\ref{item:decompOffOptOk} guarantees that, to obtain an $O(1)$-competitive procedure, it suffices to compare against the (offline) optima of the matroids $\M_i$. Combining this with property~\ref{item:decompGivesUnifDense} implies that it suffices to design a constant-competitive algorithm for uniformly dense matroids.
Since uniformly dense matroids behave in many ways very similarly to uniform matroids, which are a special case of uniformly dense matroids, it turns out that MSP on uniformly dense matroids admits a simple yet elegant $O(1)$-competitive algorithm. (See~\cite{soto_2013_matroid} for details.)
\section{Outline of Our Approach}\label{sec:outline}

As discussed, prior approaches~\cite{soto_2013_matroid,oveis_2013_variants} for RA-MSP heavily rely on knowing the matroid upfront, as they need to construct its principal sequence upfront. A natural approach would be to observe a sample set $S \subseteq N$ containing a constant fraction of all elements and then try to mimic the existing approaches using the principal sequence of $\mrestrict{\M}{S}$, the matroid $\M$ restricted to the elements in $S$. A main hurdle lies in how to analyze such a procedure as the principal sequence of $\mrestrict{\M}{S}$ can differ significantly from the one of $\M$.
In particular, one can construct matroids where it is likely that there are parts whose density is underestimated by a super-constant factor.
Moreover, $\mrestrict{\M}{S}$ may have many different densities not present in $\M$ (e.g., when $\M$ is uniformly dense).

We overcome these issues by not dealing with principal sequences directly, but rather using what we call the \emph{rank-density curve} of a matroid, which captures certain key parameters of the principal sequence. As we show, rank-density curves have three useful properties:
\begin{enumerate}
    \item They provide a natural way to derive a quantity that both relates to the offline optimum and can be easily compared against to bound the competitiveness of our procedure.
    \item They can be learned approximately by observing an $O(1)$-fraction of $N$.
    \item Approximate rank-density curves can be used algorithmically to protect denser areas from sparser ones without having to know the matroid upfront.
\end{enumerate}
\cref{sec:rankdensity} introduces rank-density curves and shows how they conveniently allow for deriving a quantity that compares against the offline optimum. \cref{sec:proofStrategy} then discusses our results on approximately learning rank-density curves and how this can be exploited algorithmically.

\subsection{Rank-Density Curves}\label{sec:rankdensity}

Given a matroid $\M = (N, \I)$, one natural way to define its rank-density curve $\rho_{\M} \colon \R_{> 0} \to \R_{\geq 0}$ is through its principal minors $\M_{1}, \dots, \M_{k}$, which are defined through the principal sequence $\emptyset \subsetneq S_{1} \subsetneq \dots \subsetneq S_{k} = N$ as explained in \cref{sec:densities}. For a value $t \in (0, \rank(\M)]$, let $i_{t} \in [k]$ be the smallest index such that $r(S_{i_{t}}) \geq t$. The value $\rho_{\M}(t)$ is then given by the density of $\M_{i_{t}}$. (See \cref{subfig:rankDensityCurve} for an example.) In addition, we set $\rho_{\M}(t) = 0$ for any $t > \rank(\M)$.

A formally equivalent way to define $\rho_{\M}$, which is more convenient for what we do later, is as follows. For any $S \subseteq N$ and $\lambda \in \R_{\geq 0}$, we define
\begin{equation}\label{eq:DSLambda}
    D_{\M}(S, \lambda) \in \underset{U \subseteq S}{\argmax}\left\{\card{U} - \lambda r(U)\right\}
\end{equation}
to be the unique maximal maximizer of $\max_{U \subseteq S}\{\card{U} - \lambda r(U)\}$. It is well-known that each set in the principal sequence $S_{1}, \dots, S_{k}$ is nonempty and of the form $D_{\M}(N, \lambda)$ for $\lambda \in \R_{\geq 0}$. This leads to the following way to define the rank-density curve, which is the one we use in what follows.

\begin{definition}[rank-density curve]\label{def:curve}
    Let $\M = (N, \I)$ be a matroid. Its \emph{rank-density} curve $\rho_{\M} \colon \R_{> 0} \to \R_{\geq 0}$ is defined by
    \begin{equation*}
        \rho_{\M}(t) \coloneqq \begin{cases}
            \max\left\{\lambda \in \R_{\geq 0} \colon r(D_{\M}(N, \lambda)) \geq t\right\} & \forall t \in (0, \rank(\M)] \\
            0                                                                              & \forall t > \rank(\M).
        \end{cases}
    \end{equation*}
\end{definition}

When the matroid $\M$ is clear from context, we also simply write $\rho$ instead of $\rho_{\M}$ for its rank-density curve and $D(N, \lambda)$ instead of $D_{\M}(N, \lambda)$. Note that $\rho$ is piecewise constant, left-continuous, and non-increasing. (See \cref{subfig:rankDensityCurve} for an example.) If $\M$ is a uniformly dense matroid with density $\lambda$, we have $\rho(t) = \lambda$ for $t\in (0,\rank(\mathcal{M})]$ and $\rho(t) = 0$ for $t\in (\rank(\mathcal{M}),\infty)$.

We now expand on how $\rho_{\M}$ is related to the expected offline optimum value $\E[\OPT(\M)]$ of an RA-MSP instance. To this end, we use the function $\eta \colon [0, \card{N}] \to \R_{\geq 0}$  defined by
\begin{equation}\label{eq:funcEta}
    \eta(a) \coloneqq \E_{R \sim \UnifDistr(N, \lfloor a \rfloor)}\left[\max_{e \in R} w(e)\right],
\end{equation}
where $\UnifDistr(N, \lfloor a \rfloor)$ is a uniformly random set of $\lfloor a \rfloor$ many elements out of $N$ (without repetitions); and we set $\eta(a)=0$ for $a \in [0,1)$ (i.e., when the set $R$ above is empty) by convention. 
In words, $\eta(a)$ is the expected maximum weight out of $\lfloor a \rfloor$ weights chosen uniformly at random from all the weights $\{w_{e}\}_{e \in N}$. Based on this notion, we assign the following value $F(\rho)$ to a rank-density curve $\rho$:
\begin{equation}\label{eq:funcF}
    F(\rho) \coloneqq \int_{0}^{\infty} \eta(\rho(t))\, dt = \int_{0}^{\rank(M)} \eta(\rho(t))\, dt,
\end{equation}
where the second equality holds because $\rho(t) = 0$ for $t> \rank(\mathcal{M})$.
Note that as the graph of $\rho$ is a staircase, the above integral is just a finite sum. 

One key property of the function $F$, is that it can be used as a proxy for the expected value of the offline optimum.
More precisely, %
the statement below shows that $F(\rho)$ is at most a constant factor smaller than the offline optimum --- see \cref{sec:curveProperties} for proof details.
This is the direction we need to be able to compare the output of our algorithm against the offline optimum.
Moreover, $F(\rho)$ is also no more than a constant-factor larger, which we do not need for our derivations; however, this is in particular also a consequence of the fact that our algorithm returns an independent set of expected weight $\Omega(F(\rho))$.
$\Cref{lem:optCurveRelation}$ is phrased in a slightly more general form that allows for applying it not just to the original matroid, but also to any minors thereof, which we need later.

\begin{lemma}\label{lem:optCurveRelation}
Let $w_1,\dots, w_n\in \mathbb{R}_{\geq 0}$ be $n$ weights, and let $\M$ be a matroid with a ground set of size $k\leq n$. 
Assume we first choose a uniformly random subset of $k$ weights among $w_1,\dots, w_n$, and then assign these weights uniformly at random to the elements of $\M$.
Then $\E[w(\OPT(\M))] \leq \frac{3e}{e - 1} \cdot F(\rho_{\M})$.
\end{lemma}

Thus, to be constant-competitive, it suffices to provide an algorithm returning an independent set of expected weight $\Omega(F(\rho))$.

\paragraph{RA-MSP Subinstances}

We will often work with minors of the matroid that is originally given in our RA-MSP instance, and apply certain results to such minors instead of the original matroid.
To avoid confusion, we fix throughout the paper one RA-MSP instance with matroid $\M_{\orig} = (N_{\orig}, \I_{\orig})$, whose ground set size we denote by $n\coloneqq |N_{\orig}|$, and whose elements have unknown but (adversarially) fixed weights $w \colon N_{\orig} \to \R_{\geq 0}$, and our goal is to design an $O(1)$-competitive algorithm for this one instance.
The weights $w$ of the original instance are the only weights we consider, even when working with RA-MSP subinstances on minors of $\M_{\orig}$, as their elements also obtain their weights uniformly at random from $w$.
In particular, the function $F$ as defined in~\eqref{eq:funcF} is always defined with respect to the original vector of $n$ weights $w$.

To formally describe the type of matroids we get as subinstances, we introduce the notion of a \emph{matroid with $w$-sampled weights}.
More precisely, if $\mathcal{M}$ is a matroid with a ground set size of $k\leq n$, then \emph{$\mathcal{M}$ with $w$-sampled weights} is a randomly weighted version of the matroid, where we pick a uniform subset of $k$ among the $n$ entries of $w$ and assign them uniformly at random to the ground set of $\mathcal{M}$.
Clearly, any minor of $\mathcal{M}$ is of this type.

Even though we may have $k<n$, a matroid $\M$ with $w$-sampled weights can be interpreted as a RA-MSP instance, as it corresponds to the adversary first choosing uniformly at random a subset of $k$ weights among the weights in $w$, which then get assigned uniformly at random to the elements.

\subsection{Proof Plan for \cref{thm:main} via Rank-Density Curves}\label{sec:proofStrategy}

We now expand on how one can learn an approximation $\rhoapprox$ of the rank-density curve $\rho_{\M_{\orig}}$ and how this can be exploited algorithmically to return an independent set of expected weight $\Omega(F(\rho_{\M_{\orig}}))$, which by \cref{lem:optCurveRelation} implies $O(1)$-competitiveness of the procedure.
To this end, we start by formalizing the notion of an \emph{approximate rank-density} curve, which relies on the notion of \emph{downshift}.

\begin{definition}\label{def:approxCurve}
    Let $\rho \colon \R_{> 0} \to \R_{\geq 0}$ be a non-increasing function and let $\alpha, \beta \in \R_{\geq 1}$. The \emph{$(\alpha, \beta)$-downshift} $\rhodownshift \colon \R_{> 0} \to \R_{\geq 0}$ of $\rho$ is defined via an auxiliary function $\phi \colon \R_{> 0} \to \R_{\geq 0}$ as follows:
    \begin{equation*}
        \phi(t) \coloneqq \begin{cases}
            \frac{\rho(\alpha)}{\beta}        & \forall t \in (0, 1], \\
            \frac{\rho(\alpha\cdot t)}{\beta} & \forall t > 1;
        \end{cases}
        \quad
        \rhodownshift(t) \coloneqq \begin{cases}
            1 & \text{ if } \phi(t) \in (0, 1), \\
            \phi(t) & \text{ otherwise }.
        \end{cases}
    \end{equation*}
    Moreover, a function $\rhoapprox \colon \R_{> 0} \to \R_{\geq 0}$ is called an \emph{$(\alpha, \beta)$-approximation} of $\rho$ if it is non-increasing and $\rhodownshift \leq \rhoapprox \leq \rho$, where $\rhodownshift$ is the $(\alpha, \beta)$-downshift of $\rho$.
\end{definition}

One helpful way to think about an $(\alpha,\beta)$-downshift is as a slightly modified version of $\sfrac{\rho(\alpha\cdot t)}{\beta}$.
This is also where the name stems from, as $\sfrac{\rho(\alpha\cdot t)}{\beta}$ corresponds to shifting, when thinking in doubly logarithmic scale, the function $\rho$ to the left and down, corresponding to the factors $\alpha$ and $\beta$, correspondingly.
This function is then modified in two ways.
First, for $t\in (0,1]$, we lower its value to $\sfrac{\alpha}{\beta}$.
This is done because we are not able to accurately estimate densities for low ranks.
Fortunately, this turns out not to be an issue to obtain a constant-competitive algorithm because we can set off the loss of this modification by running, with some probability, the classical single secretary algorithm to return the heaviest element with constant probability.
(This is in particular implied by \Cref{lem:shiftedCurve} below and discussed right after.)
The second modification is that we round up values in $(0,1)$.
This reflects the fact that density values are always at least one.

One issue when working with an $(O(1), O(1))$-approximation $\rhoapprox$ of $\rho$ is that $F(\rhoapprox)$ may be more than a constant factor smaller than $F(\rho)$, and we thus cannot compare against $F(\rhoapprox)$ to obtain an $O(1)$-competitive procedure.
This happens due to the above-mentioned way how $(\alpha,\beta)$-downshifts are defined; more precisely, that values for $t\in (0,1]$ got rounded down to $\sfrac{\rho(\alpha)}{\beta}$.
However, as the following lemma shows, also in this case we can obtain a simple lower bound for the value $F(\rhoapprox)$ in terms of $F(\rho)$ and the largest weight $w_{\max}$ in $w$ --- a proof of the statement can be found at the end of \Cref{sec:curveProperties}.

\begin{lemma}\label{lem:shiftedCurve}
    Let $\M$ be a matroid with $w$-sampled weights, let $\alpha,\beta \in \mathbb{R}_{\geq 1}$, and let $\rhoapprox$ be an $(\alpha, \beta)$-approximation of $\rho=\rho_{\M}$. Then $F(\rho) \leq 2 \alpha \beta F(\rhoapprox) + \alpha w_{\max}$.
\end{lemma}

A key implication of \cref{lem:shiftedCurve} is that it suffices to obtain an algorithm that returns an independent set of expected weight $\Omega(F(\rhoapprox))$ for some $(O(1), O(1))$-approximation $\rhoapprox$ of $\rho_{\M_{\orig}}$. Indeed, \cref{lem:shiftedCurve} then implies $F(\rhoapprox) = \Omega(F(\rho_{\M_{\orig}})) - O(w_{\max})$. By running this algorithm with some probability (say $0.5$) and otherwise \citeauthor{dynkin_1963_secretary}'s~\cite{dynkin_1963_secretary} classical secretary algorithm, which picks the heaviest element with constant probability, an overall algorithm is obtained that returns an independent set of expected weight $\Omega(F(\rho_{\M_{\orig}}))$. Hence, \cref{lem:shiftedCurve} helps to provide bounds on the competitiveness of algorithms that are competitive with the $F$-value of an approximate rank-density curve. This technique is also used in the following key statement, which shows that an algorithm with strong guarantees can be obtained if we are given an $(O(1), O(1))$-approximation of the rank-density curve of the matroid on which we work --- see \cref{sec:algorithm} for the proof.

\begin{theorem}\label{thm:aidedAlgorithm}
	Let $\M$ be a matroid with $w$-sampled weights, and let $\rho_{\M}$ denote the rank-density curve of $\M$.
    Assume we are given an $(\alpha, \beta)$-approximation $\rhoapprox$ of $\rho_{\M}$ for integers $\alpha \geq 24$ and $\beta \geq 3$. Then there is an efficient procedure $\mathrm{ALG}(\rhoapprox, \alpha, \beta)$ that, when run on the RA-MSP subinstance given by $\M$, returns an independent set $I$ of $\M$ of expected weight at least $\left(\tfrac{1}{1440 e \alpha^{2} \beta^{2}}\right) \left(F(\rho_{\M})- \alpha^{2} w_{\max} \right)$.
\end{theorem}

The last main ingredient of our approach is to show that such an \emph{accurate} proxy $\rhoapprox$ can be computed with constant probability. More precisely, we show that, after observing a sample set $S$ containing every element of $N_{\orig}$ independently with probability $\sfrac{1}{2}$, the rank-density curve of (the observed) $\mrestrict{\M_{\orig}}{S}$ 
\begin{itemize}
\item is close to the rank-density curve of $\mrestrict{\M_{\orig}}{N_{\orig} \setminus S}$, allowing us to use $\rho_{\mrestrict{\M_{\orig}}{S}}$ as desired proxy for the RA-MSP subinstance given by $\mrestrict{\M_{\orig}}{N_{\orig} \setminus S}$, and
\item is close to the rank-density curve of $\M_{\orig}$, which allows for relating the offline optimum of the RA-MSP subinstance given by $\mrestrict{\M_{\orig}}{N_{\orig} \setminus S}$ to the one of $\M_{\orig}$.
\end{itemize}
We highlight that the next result is purely structural and hence independent of weights or the MSP setting. See \cref{sec:curveEstimate} for details.

\begin{theorem}\label{thm:goodEvent}
	Let $\M=(N, \I)$ be a matroid and $S \subseteq N$ be a random set containing every element of $N$ independently with probability $\sfrac{1}{2}$. Then, with probability at least $\sfrac{1}{100}$, $\rho_{\mrestrict{\M}{S}}$ and $\rho_{\mrestrict{\M}{N \setminus S}}$ are both $(288, 9)$-approximations of $\rho_{\M}$.
\end{theorem}

Combining the above results, we get the desired $O(1)$-competitive algorithm.
\begin{proof}[Proof of \cref{thm:main}]
    For brevity, let $\M \coloneqq \M_{\orig}$ and $N \coloneqq N_{\orig}$ throughout this proof. 
    Recall that by \cref{lem:optCurveRelation}, it suffices to provide an algorithm returning an independent set of expected weight $\Omega(F(\rho_{\M}))$.
    Consider the following procedure: First observe (without picking any element) a set $S \subseteq N$ containing every element of $N$ independently with probability $\sfrac{1}{2}$ and let $\rhoapprox$ denote the $(288, 9)$-downshift of $\rho_{\mrestrict{\M}{S}}$. Then run the algorithm described in \cref{thm:aidedAlgorithm} on $\mrestrict{\M}{N \setminus S}$ with $\rhoapprox$ as the approximate rank-density curve. Let $I$ denote the output of the above procedure and let $\mathcal{A}$ be the event defined in \cref{thm:goodEvent}, that is,
	\[
	    \mathcal{A} = \{S \subseteq N \colon \rho_{\mrestrict{\M}{S}} \text{ and } \rho_{\mrestrict{\M}{N\setminus S}} \text{ are } (288, 9)\text{-approximations of } \rho_{\M}\}.
	\]
A key property we exploit is that for any $S\in \mathcal{A}$, we have that $\tilde{\rho}$ is a $(288^2, 9^2)$-approximation of $\rho_{\mrestrict{\M}{N\setminus S}}$ due to the following.
First, because $\rhoapprox$ is the $(288, 9)$-downshift of $\rho_{\mrestrict{\M}{S}} \leq \rho_{\M}$, and $\rho_{\mrestrict{\M}{N \setminus S}}$ is a $(288, 9)$-approximations of $\rho_{\M}$, we have $\rho_{\mrestrict{\M}{N \setminus S}} \geq \rhoapprox$.
Moreover, the approximation parameter $(288^{2}, 9^{2})$ follows by using that the $(\alpha_{2}, \beta_{2})$-downshift of the $(\alpha_{1}, \beta_{1})$-downshift of some rank-density function is an $(\alpha_{1} \alpha_{2}, \beta_{1} \beta_{2})$-approximation of that rank-density function --- see \cref{lem:approxOfApprox} for a proof of this property.
This property can be applied as follows.
Let $\bar{\rho}$ be the $(288,9)$-downshift of $\rho_{\mrestrict{\M}{N\setminus S}}\leq \rho_{\M}$.
Because $\rho_{\mrestrict{\M}{S}}$ is a $(288,9)$-approximation of $\rho_{\M}$, we have $\bar{\rho}\leq \rho_{\mrestrict{\M}{S}}$.
Hence, because $\tilde{\rho}$ is the $(288,9)$-downshift of $\rho_{\mrestrict{\M}{S}}$, it lies above the $(288,9)$-downshift of $\bar{\rho}$, which is itself the $(288,9)$-downshift of $\rho_{\mrestrict{\M}{N\setminus S}}$.
Thus, by the above property, we obtain that $\tilde{\rho}$ is a $(288^2,9^2)$-approximation of $\rho_{\mrestrict{\M}{N\setminus S}}$ as claimed.

Using this fact, we obtain for any fixed $S\in \mathcal{A}$
\begin{align*}
\E[w(I) \mid S] 
		&\geq \left(\frac{1}{1440 e \cdot 288^{4} \cdot 9^{4}}\right) \left(F \left( \rho_{\mrestrict{\M}{N\setminus S}} \right)- 288^{4} w_{\max} \right) \\
		&\geq \left(\frac{1}{1440 e \cdot 288^{4} \cdot 9^{4}}\right) \left( \frac{1}{2\cdot 288\cdot 9}(F \left( \rho_{\M} \right) - 288 w_{\max}) - 288^{4} w_{\max} \right) \\
		&\geq \left(\frac{1}{1440 e \cdot 288^{4} \cdot 9^{4}}\right) \left( \frac{1}{2\cdot 288\cdot 9}F \left( \rho_{\M} \right) - 2\cdot 288^{4} w_{\max} \right) \\
		&= \left(\frac{1}{2880 e \cdot 288^{5} \cdot 9^{5}}\right) F \left( \rho_{\M} \right)  - \frac{w_{\max}}{720 e \cdot 9^4},
\end{align*}
	where the first inequality follows from \cref{thm:aidedAlgorithm} and the fact that $\tilde{\rho}$ is a $(288^2,9^2)$-approximation of $\rho_{\mrestrict{\M}{N\setminus S}}$ as discussion above, while the second inequality follows from \cref{lem:shiftedCurve} and the fact that, for every $S \in \mathcal{A}$, the curve $\rho_{\mrestrict{\M}{N \setminus S}}$ is a $(288, 9)$-approximation of $\rho_{\M}$. 
	Moreover, the first inequality uses that conditioning on any fixed $S \in \mathcal{A}$ does not have any impact on the uniform assignment of the weights $w$ to the elements. This holds because the event $\mathcal{A}$ only depends on the sampled elements $S$ but not the weights of its elements. 
	Hence, the RA-MSP subinstance given by $\mrestrict{\M}{N\setminus S}$ on which we use the algorithm described in \cref{thm:aidedAlgorithm} indeed assigns weights of $w$ uniformly at random to elements, as required.
	It then follows that the output of the above procedure satisfies
    \begin{align*}
    	\E[w(I)]
    	    & \geq \sum_{S \in \mathcal{A}} \E[w(I) \mid S] \Pr[S]
    	      \geq  \tfrac{1}{100} \left( \left(\tfrac{1}{2880 e \cdot 288^{5} \cdot 9^{5}}\right) F \left( \rho_{\M} \right)  - \tfrac{w_{\max}}{720e \cdot 9^4} \right),
    \end{align*}
	where the last inequality uses that $\P[\mathcal{A}] \geq \sfrac{1}{100}$ by \cref{thm:goodEvent}.
    
    Since running the classical secretary algorithm on $\M_{\orig}$ returns an independent set of expected weight at least $\sfrac{w_{\max}}{e}$, the desired result now follows by running the procedure described above with probability $\sfrac{1}{2}$, and running the classical secretary algorithm otherwise.
\end{proof}
\section{Rank-Density Curves and Their Properties}\label{sec:curveProperties}

In this section we take a closer look at rank-density curves, and prove \cref{lem:optCurveRelation,lem:shiftedCurve}. We start by stating some useful properties of the function $\eta$.

\begin{lemma}\label{lem:etaProperties}
	Let $\eta$ be as defined in \eqref{eq:funcEta}. Then
	\begin{enumerate}
		\item\label{item:etaIncreasing} $\eta$ is non-decreasing.
		\item\label{item:etaProduct} $\eta(ah) \leq 2a \eta(h)$ for all $a \geq 1$ and $h \in [1, \frac{n}{a}]$.
		\item\label{item:etaExpectation} Let $X \sim \BinomDistr(m, p)$ with $1 \leq mp \leq m \leq n$. Then $\E[\eta(X)] \leq 3 \eta(mp)$.
	\end{enumerate}
\end{lemma}

\begin{proof}
	\begin{enumerate}
		\item[\ref{item:etaIncreasing}] This follows immediately from the definition of $\eta$.
		\item[\ref{item:etaProduct}] Consider a consecutive numbering of the $n$ entries of the weight vector $w\in \mathbb{R}^n$ in non-decreasing order, i.e., $w_1 \leq w_2 \leq \cdots \leq w_n$.
For each $k, i \in [n]$ let $p(k, i)$ denote the probability that among $k$ samples, element $i$ is the heaviest and has the smallest index among the heaviest elements (in case of ties). In other words, $p(k, i)$ is the probability that $w_{i}$ is the heaviest weight out of $k$ sampled ones and none of $w_{j}$ with $j < i$ are in the sample. Thus for $i < k$ we have $p(k, i) = 0$ and for $i \geq k$ we have
        \[
            p(k , i) = \frac{\binom{i - 1}{k - 1}}{\binom{n}{k}} = k \frac{(i - 1)! (n - k)!}{n! (i - k)!} = \frac{k}{n} \prod_{j = 1}^{k - 1} \frac{i - j}{n - j}.
        \]
        Next, note that for all $h \in [1, \frac{n}{a}]$ we have
\begin{align*}
\eta(h)  &= \sum_{i = 1}^{n} p(\lfloor h \rfloor, i) w_{i}, \text{ and}\\
\eta(ah) &= \sum_{i = 1}^{n} p(\lfloor ah \rfloor, i) w_{i}.
\end{align*}
Our goal is to show that $p(\lfloor ah \rfloor, i) \leq 2 a p(\lfloor h \rfloor, i)$ for all $i \in [n]$, which is sufficient to prove the desired inequality. To this end, note that for $i < \lfloor ah \rfloor$ we have $p(\lfloor ah \rfloor, i) = 0 \leq a p(\lfloor h \rfloor, i)$, and for $i \geq \lfloor ah \rfloor$ we have
        \begin{align*}
            p(\lfloor ah \rfloor, i)
                &   =  \frac{\lfloor ah \rfloor}{n} \prod_{j = 1}^{\lfloor ah \rfloor - 1} \frac{i - j}{n - j}
                  \leq \frac{\lfloor ah \rfloor}{n} \prod_{j = 1}^{\lfloor h \rfloor - 1}  \frac{i - j}{n - j}
                  \leq 2a \cdot \frac{\lfloor h \rfloor}{n} \prod_{j = 1}^{\lfloor h \rfloor - 1} \frac{i - j}{n - j}
                    =  2a p(\lfloor h \rfloor, i).
        \end{align*}
        Here the first inequality follows from $\lfloor h \rfloor \leq \lfloor ah \rfloor$ and the fact that $\frac{i - j}{n - j} \leq 1$ for any $j < n$, and the second inequality uses that $\lfloor ah \rfloor \leq ah \leq 2a \lfloor h \rfloor$ because $h \leq 2 \lfloor h \rfloor$ for $h \in \mathbb{R}_{\geq 1}$.
		
		\item[\ref{item:etaExpectation}] Note that $\eta(h) \leq \sfrac{2h}{mp} \cdot \eta(mp)$ for all $h \in [mp, n]$ by property \ref{item:etaProduct}. Therefore,
		\begin{align*}
			\E\left[\eta(X)\right]
			& = \sum_{h = 0}^{\lfloor mp \rfloor} \eta(h) \cdot \Pr\left[X = h\right] + \sum_{h = \lfloor mp \rfloor + 1}^{m} \eta(h) \cdot \Pr\left[X = h\right] \\
			& \leq \eta(mp) \cdot \Pr\left[X \leq \lfloor mp \rfloor\right] + \sum_{h = \lfloor mp \rfloor + 1}^{m} 2 \frac{h}{mp} \cdot \eta(mp) \cdot \Pr\left[X = h\right] \\
			& = \eta(mp) \cdot \Pr\left[X \leq \lfloor mp \rfloor \right] + 2 \frac{\eta(mp)}{mp} \cdot \sum_{h = \lfloor mp \rfloor + 1}^{m} h \cdot \Pr\left[X = h\right] \\
			& \leq \eta(mp) + 2 \frac{\eta(mp)}{mp} \cdot \E\left[X\right] = 3\eta(mp).
		\end{align*}
	\end{enumerate}
\end{proof}

In order to prove \cref{lem:optCurveRelation}, we use the following result from \cite{soto_2013_matroid}, which relies on the notion of the \emph{random partition matroid} $\mpartition = (N, \I')$ associated to a given matroid $\M = (N, \I)$. The former is constructed as follows. First, every element $e \in N$ is assigned to one of $\mathtt{rank}(\M)$ many classes $P_{1}, \dots, P_{\mathtt{rank}(\M)}$ uniformly at random and independently of each other. A set $S \subseteq N$ is then independent in $\mpartition$ if it contains at most one element from each class, i.e., $S \in \I'$ if $\card{S \cap P_{i}} \leq 1$ for every $i \in [\mathtt{rank}(\M)]$. The next result relates the value of the offline OPT of $\M$, with that of the random partition matroids associated to the principal minors of $\M$.

\begin{lemma}[{\cite[Lemma~4.2]{soto_2013_matroid}}]\label{lem:OPTpartitions}
	Let $(\M, w)$ be a random-assignment MSP instance. Let $\{\M_{i}\}_{i = 1}^{k}$ be the principal minors of $\M$ and let $\mpartition_{i}$ denote the random partition matroid associated to each $\M_{i}$, respectively. Then
	\[
	   \E\left[w(\OPT(\M))\right] \leq \frac{e}{e - 1} \E\left[w(\OPT(\oplus_{i = 1}^{k} \mpartition_{i}))\right],
	\]
	where the expectations are taken with respect to the random weight assignment and the random partitioning (this only applies to the right-hand side expectation).
\end{lemma}

Now we are ready to prove \cref{lem:optCurveRelation}.

\begin{proof}[Proof of \cref{lem:optCurveRelation}]
	Let $(\M'_{i})_{i = 1}^{k}$ be the principal minors of $\M'$ and let $n'_{i}$ and $r'_{i}$ denote the cardinality and rank of each $\M'_{i}$, respectively. Additionally, for every $i \in [k]$, let $\mpartition'_{i}$ be the random partition matroid associated to $\M'_{i}$ and let $(P'_{i, j})_{j = 1}^{r'_{i}}$ denote the (random) partitions of $\mpartition'_{i}$. By \cref{lem:OPTpartitions} we have
	\[
        \E\left[w(\OPT(\M'))\right]
            \leq \frac{e}{e - 1} \E\left[w(\OPT(\oplus_{i = 1}^{k} \mpartition'_{i}))\right]
              =  \frac{e}{e - 1} \sum_{i = 1}^{k} \E\left[w(\OPT(\mpartition'_{i}))\right].
	\]
	Next, note that for every $i \in [k]$ we have
	\[
        \E[w(\OPT(\mpartition'_{i}))]
            = \sum_{j = 1}^{r'_{i}} \E\left[w(\OPT(\mrestrict{\mpartition'_{i}}{P'_{i, j}}))\right]
            = \sum_{j = 1}^{r'_{i}} \E\left[\eta_{w}(\card{P'_{i, j}})\right]
            = r'_{i} \cdot \E_{X \sim \BinomDistr(n'_{i}, \, 1 / r'_{i})} [\eta_{w}(X)],
	\]
	where the first equality follows from the definition of $\mpartition'_{i}$ and linearity of expectation, the second one uses the definition of $\eta_{w}$, and the third one holds since $\card{P'_{i, j}} \sim \BinomDistr(n'_{i}, 1 / r'_{i})$ for every $j \in [r'_{i}]$. Thus, by applying \cref{lem:etaProperties} and using that by construction every $\M'_{i}$ is uniformly dense with density $\lambda'_{i} = n'_{i} / r'_{i}$, we get:
	\[
        \E[w(\OPT(\mpartition'_{i}))] \leq r'_{i} \cdot 3\eta(n'_{i} / r'_{i}) = 3 r'_{i} \eta(\lambda'_{i}).
	\]
	Hence
	\[
        \E\left[w(\OPT(\M'))\right]
            \leq \frac{3e}{e - 1} \sum_{i = 1}^{k} r'_{i} \eta(\lambda'_{i})
              =  \frac{3e}{e - 1} F(\rho_{\M'}),
	\]
	where the equality holds by definition of $F_{w}$.
\end{proof}

The following result shows that, simply put, in terms of \cref{def:approxCurve}, an approximation of an approximation of a function is also an approximation of the original function, where the approximation parameters $\alpha$ and $\beta$ are multiplied.

\begin{lemma}\label{lem:approxOfApprox}
    Let $\rho_{1}, \rho_{2}, \rho_{3} \colon \R_{> 0} \to \R_{\geq 0}$ be non-increasing functions such that $\rho_{2}$ is an $(\alpha_{1}, \beta_{1})$-approximation of $\rho_{1}$ and $\rho_{3}$ is an $(\alpha_{2}, \beta_{2})$-approximation of $\rho_{2}$ for some parameters $\alpha_{1}, \beta_{1}, \alpha_{2}, \beta_{2} \in \R_{\geq 1}$. Then $\rho_{3}$ is an $(\alpha_{1} \alpha_{2}, \beta_{1} \beta_{2})$-approximation of $\rho_{1}$.
\end{lemma}

\begin{proof}
    Let $\rho_{1}'$ be the $(\alpha_{1}, \beta_{1})$-downshift of $\rho_{1}$ and let $\phi_{1}$ be the auxiliary function used to construct $\rho_{1}'$ in \cref{def:approxCurve}. Similarly, let $\rho_{2}'$ and $\phi_{2}$ be the $(\alpha_{2}, \beta_{2})$-downshift of $\rho_{2}$ and the corresponding auxiliary function, respectively, and let $\rho_{1}''$ and $\phi_{1}'$ be the $(\alpha_{1} \alpha_{2}, \beta_{1} \beta_{2})$-downshift of $\rho_{1}$ and the corresponding auxiliary function, respectively.

    Since $\rho_{3} \leq \rho_{2} \leq \rho_{1}$, it only remains to show that $\rho_{3} \geq \rho_{1}''$. Observe that to obtain this bound it suffices to prove the following two properties (for a function $g$ taking real values, we denote by $\supp(g)$ all points $t$ in its domain for which $g(t) \neq 0$):
    \begin{enumerate}
        \item\label{item:AOAaux} $\rho_{3}(t) \geq \phi_{1}'(t)$ for every $t > 0$, and
        \item\label{item:AOAone} $\rho_{3}(t) \geq 1$ for every $t \in \supp(\rho_{1}'')$.
    \end{enumerate}
    Indeed, on the one hand, for every $t > 0$ such that $\phi_{1}'(t) = \rho_{1}''(t)$ the first property implies $\rho_{3}(t) \geq \rho_{1}''(t)$. On the other hand, for every $t > 0$ such that $\phi_{1}'(t) \neq \rho_{1}''(t)$ we have $\rho_{1}''(t) = 1$, so $\rho_{3}(t) \geq \rho_{1}''(t)$ follows from the second property.

    First, we prove property \ref{item:AOAaux}. To this end, observe that for every $t > 0$ we have $\rho_{3}(t) \geq \rho_{2}'(t) \geq \phi_{2}(t)$ and $\rho_{2}(t) \geq \rho_{1}'(t) \geq \phi_{1}(t)$. Therefore,
    \begin{align*}
        \rho_{3}(t)
            & \geq \phi_{2}(t)
                =  \begin{cases}
                       \frac{\rho_{2}(\alpha_{2})}{\beta_{2}} & \forall t \in (0, 1), \\
                       \frac{\rho_{2}(\alpha_{2} t)}{\beta_{2}} & \forall t \geq 1
                   \end{cases}
              \geq \begin{cases}
                       \frac{\phi_{1}(\alpha_{2})}{\beta_{2}} & \forall t \in (0, 1), \\
                       \frac{\phi_{1}(\alpha_{2} t)}{\beta_{2}} & \forall t \geq 1
                   \end{cases}
                =  \begin{cases}
                       \frac{\rho_{1}(\alpha_{1} \alpha_{2})}{\beta_{1} \beta_{2}} & \forall t \in (0, 1), \\
                       \frac{\rho_{1}(\alpha_{1} \alpha_{2} t)}{\beta_{1} \beta_{2}} & \forall t \geq 1
                   \end{cases}
                =  \phi_{1}'(t).
    \end{align*}
    Here the first and second inequalities hold by the above observation, the first and the last equalities hold by construction of $\phi_{2}$ and $\phi_{1}'$, respectively, and the second equality holds by construction of $\phi_{1}$ and the fact that $\alpha_{2} \geq 1$.

    Now, let us show property \ref{item:AOAone}. First, note that for every $t \in \supp(\rho_{2}')$ by construction we have $\rho_{3}(t) \geq \rho_{2}'(t) \geq 1$. Next, note that by construction
    \[
        \supp(\rho_{2}') = \supp(\phi_{2}) = \begin{cases}
            \emptyset & \text{ if } \alpha_{2} \notin \supp(\rho_{2}) \\
            \{t > 0 \colon \alpha_{2} t \in \supp(\rho_{2})\} & \text{ otherwise }.
        \end{cases}
    \]
    Since $\supp(\rho_{2}) \supseteq \supp(\rho_{1}')$, this implies
    \[
        \supp(\rho_{2}') \supseteq \begin{cases}
            \emptyset & \text{ if } \alpha_{2} \notin \supp(\rho_{1}') \\
            \{t > 0 \colon \alpha_{2} t \in \supp(\rho_{1}')\} & \text{ otherwise }.
        \end{cases}
    \]
    Similarly, note that
    \[
        \supp(\rho_{1}') = \supp(\phi_{1}) = \begin{cases}
            \emptyset & \text{ if } \alpha_{1} \notin \supp(\rho_{1}) \\
            \{t > 0 \colon \alpha_{1} t \in \supp(\rho_{1})\} & \text{ otherwise }.
        \end{cases} \\
    \]
    Thus, the condition $\alpha_{2} \notin \supp(\rho_{1}')$ is equivalent to the following: either $\alpha_{1} \notin \supp(\rho_{1})$ or $\alpha_{1} \alpha_{2} t \notin \supp(\rho_{1})$. Note that the latter condition is more restrictive, as $\alpha_{2} \geq 1$. Therefore, we get:
    \[
        \supp(\rho_{2}') \supseteq \begin{cases}
            \emptyset & \text{ if } \alpha_{1} \alpha_{2} \notin \supp(\rho_{1}) \\
            \{t > 0 \colon \alpha_{1} \alpha_{2} t \in \supp(\rho_{1})\} & \text{ otherwise }
        \end{cases} = \supp(\rho_{1}''),
    \]
    where the equality follows by construction of $\supp(\rho_{1}'')$. Thus, $\rho_{3}(t) \geq 1$ holds for every $t \in \supp(\rho_{2}') \supseteq \supp(\rho_{1}'')$, so property \ref{item:AOAone} holds.
\end{proof}

We conclude this section by providing a proof of \Cref{lem:shiftedCurve}.

\begin{proof}[Proof of \Cref{lem:shiftedCurve}]
First observe that the statement trivially holds if we have $\alpha \geq \rank(\mathcal{M})$, because $F(\rho) \leq \alpha w_{\max}$.
Hence, in what follows we assume $\alpha < \rank(\mathcal{M})$.

Let $\rho'\colon \mathbb{R}_{>0} \to \mathbb{R}_{\geq 0}$ be the $(\alpha,\beta)$-downshift of $\rho$.
By definition of the function $F$ and because $\tilde{\rho} \geq \rho'$, we get
\begin{equation}\label{eq:FTildeToDownshift}
F(\tilde{\rho}) = \int_0^{\infty} \eta(\tilde{\rho}(t))\,dt
 \geq \int_0^{\infty} \eta(\rho'(t))\,dt
 = \int_0^{\frac{\rank(\mathcal{M})}{\alpha}} \eta(\rho'(t))\,dt,
\end{equation}
where the equality at the end follows from the fact that $\rho'(t)=0$ for $t > \sfrac{\rank(\mathcal{M})}{\alpha}$.
We now distinguish between $\beta \geq n$ and $\beta < n$.

Consider first the case $\beta \geq n$.
In this case we can expand as follows:
\begin{align*}
\int_0^{\frac{\rank(\mathcal{M})}{\alpha}} \eta(\rho'(t))\,dt
&\geq \frac{\rank(\mathcal{M})}{\alpha} \cdot \eta(1)
 \geq \frac{1}{2\alpha\beta} \rank(\mathcal{M}) \cdot 2n \eta(1)
 \geq \frac{1}{2\alpha\beta} \rank(\mathcal{M})\cdot \eta(n)\\
&= \frac{1}{2\alpha\beta} \rank(\mathcal{M})\cdot w_{\max}
\geq \frac{1}{2\alpha\beta} F(\rho),
\end{align*}
where the first inequality holds because $\rho'$ is at least $1$ when it is non-zero and $\rho'$ is non-zero within $[0,\sfrac{\rank(\M)}{\alpha}]$ because $\alpha\leq \rank(\M)$, the second inequality uses $\beta \geq n$, the third one is due to $2n\eta(1) \geq \eta(n)$---which is a consequence of \cref{lem:etaProperties}~\ref{item:etaProduct}---the equality thereafter uses $\eta(n)=w_{\max}$ and the final inequality follows from $\rank(\mathcal{M})\cdot w_{\max} \geq F(\rho)$.
The above relation together with \Cref{eq:FTildeToDownshift} implies the desired result when $\beta \geq n$.

Now assume $\beta < n$.
In this case we continue as follows:
\begin{align*}
\int_0^{\frac{\rank(\mathcal{M})}{\alpha}} \eta(\rho'(t))\,dt
  &\geq \int_1^{\frac{\rank(\mathcal{M})}{\alpha}} \eta(\rho'(t))\,dt\\
  &= \int_1^{\frac{\rank(\mathcal{M})}{\alpha}} \eta\left(\max\left\{\frac{\rho(\alpha t)}{\beta},1\right\}\right)\,dt\\
  &= \frac{1}{\alpha} \int_{\alpha}^{\rank(\mathcal{M})} \eta\left(\max\left\{\frac{\rho(t)}{\beta},1\right\}\right)\,dt\\
  &\geq \frac{1}{2 \alpha \beta} \int_{\alpha}^{\rank(\mathcal{M})} \eta\left(\max\left\{\rho(t),\beta\right\}\right)\,dt\\
  &\geq \frac{1}{2 \alpha \beta} \int_{\alpha}^{\rank(\mathcal{M})} \eta(\rho(t))\,dt\\
  &= \frac{1}{2 \alpha \beta} \left(\int_{0}^{\rank(\mathcal{M})} \eta(\rho(t))\,dt -  \int_0^{\alpha} \eta(\rho(t))\,dt\right)\\
  &\geq \frac{1}{2 \alpha \beta} (F(\rho) - \alpha w_{\max}),
\end{align*}
where the first equality is a consequence of $\rho'$ being the $(\alpha,\beta)$-downshift of $\rho$,
the second equality follows from a variable substitution,
the second inequality uses $2\beta \eta(\max\{\sfrac{\rho(t)}{\beta},1\}) \geq \eta(\max\{\rho(t),\beta\})$, which holds due to \cref{lem:etaProperties}~\ref{item:etaProduct} (here we use $\beta \leq n$ to fulfill the conditions of this statement),
and in the last inequality we use the definition of $F$ and the fact that the function $\eta$ never takes values larger than $w_{\max}$.
The above relation together with \Cref{eq:FTildeToDownshift} implies the desired result for $\beta \leq n$, which finishes the proof.
\end{proof}

\section{Learning Rank-Density Curves From a Sample}\label{sec:curveEstimate}

One of the main challenges when designing and analyzing algorithms for MSP is understanding what kind of (and how much) information can be learned about the underlying instance after observing a random sample of it. 

In this section, we show that, with constant probability, after observing a sample set $S$ containing each element with probability $0.5$, one can learn a good approximation of the rank-density curve of both $\M$ and $\mrestrict{\M}{N \setminus S}$, thus proving \cref{thm:goodEvent}. However, even if one knew the exact (instead of an approximate) rank-density curve of $\mrestrict{\M}{N \setminus S}$, given that the matroid is not known upfront (and hence neither which elements are associated to each of the different density areas of the curve), it is a priori not clear how to proceed. A second main contribution of this section is to show that the set of elements in $N \setminus S$ that are spanned by a subset of $S$ of a given density is well-structured. In particular, this will allow us to build a (chain) decomposition $\bigoplus_{i = 1}^{k} \M_{i}$ of $\mrestrict{\M}{N \setminus S}$ where all the $\M_{i}$'s satisfy some desired properties with constant probability --- see \cref{sec:proofThmGRP} for details.

The main technical contribution in this section is the following result.

\begin{theorem}\label{thm:concentrationBounds}
	Let $\M = (N, \I)$ be a matroid containing $3 h$ disjoint bases for some $h \in \Z_{\geq 1}$. Let $S \sim \BinomDistr(N, \sfrac{1}{2})$. Then
	\begin{align}
		\Pr\left[ |\mspan(D(S, h)) \setminus S| \leq \frac{\card{N}}{12}\right]
		&\leq \mathtt{exp}\left(-\frac{\card{N}}{144}\right),\label{eq:concentrationCard}\\
		\Pr\left[r(D(S, h)) \leq \frac{r(N)}{8}\right]
		&\leq \mathtt{exp}\left(-\frac{r(N)}{48}\right).\label{eq:concentrationRankS}
	\end{align}
\end{theorem}
\begin{proof}
	We prove the concentration result \eqref{eq:concentrationCard} first. Let $\M_{h} = (N, \I_{h})$ denote the $h$-fold union of $\M$ and let $r_{h}$ denote its rank function. Consider the procedure described in \cref{alg:karger-inspired}, which is loosely inspired by \cite{karger_1998}.
    \begin{algorithm}[th]
        \caption{Algorithm for lower bounding $\E_{S}\left[\card{\mspan(D(S, h)) \setminus S}\right]$}\label{alg:karger-inspired}
        Set $W \leftarrow \emptyset$, $G \leftarrow \emptyset$, and $C \leftarrow \emptyset$\;
		\For{every $e \in N$ considered in an arbitrary order}{
			\eIf{$W \cup \{e\} \in \I_{h}$}{
				\lIf{$e \in S$}{Update $W \leftarrow W \cup \{e\}$}
			}{
				\leIf{$e \in S$}{Update $C \leftarrow C \cup \{e\}$}{Update $G \leftarrow G \cup \{e\}$}
			}
		}
	\end{algorithm}
    Note that the following three properties hold at all times: $W$, $G$, and $C$ are pairwise disjoint; $W \subseteq S$ and $C \subseteq S$, while $G \cap S = \emptyset$; and $W \in \I_{h}$. In addition, by construction, at the end of the procedure we have:
	\begin{enumerate}
		\item\label{item:karger1} $S = C \uplus W$. Moreover, the random sets $G$ and $C$ have identical distributions, because each element belongs to $S$ with probability $\sfrac{1}{2}$ independently of the other elements.
		\item\label{item:karger2} $G \subseteq \mspan(D(S, h)) \setminus S$. Because $G \cap S = \emptyset$, it is enough to show $G \subseteq \mspan(D(S, h))$. Given an arbitrary $e \in G$, by construction we have $W \cup \{e\} \notin \I_{h}$, i.e., $r_{h}(W \cup \{e\}) = r_{h}(W)$. As $W \subseteq S$, this yields $r_{h}(S \cup \{e\}) = r_{h}(S)$, which then implies $e \in \mspan(D(S, h))$. The latter implication follows by a standard matroid argument; we provide a proof in \cref{lem:density-union} for completeness.
	\end{enumerate}

    \noindent As $G\subseteq \mspan(D(S,h))\setminus S$, and $G$ and $C$ have the same distribution, we get
    \begin{equation}\label{eq:boundGInsteadOfSpan}
        \Pr\left[\card{\mspan(D(S,h)\setminus S)} \leq \tfrac{\card{N}}{12}\right]
            \leq \Pr\left[\card{G} \leq \tfrac{\card{N}}{12}\right]
              =  \Pr\left[\card{C} \leq \tfrac{\card{N}}{12}\right].
    \end{equation}
    Moreover,
    \begin{equation}\label{eq:boundCInS}
        \card{C} = \card{S} - \card{W} \geq \card{S} - h r(N) \geq \card{S} - \sfrac{\card{N}}{3},
    \end{equation}
    where the equality follows from $S = C \uplus W$, the first inequality from $W \in \I_{h}$ (which implies $\card{W} = r_{h}(W) \leq h r(N)$), and the last one from the fact that $\M$ contains $3 h$ many disjoint bases (and hence $\card{N} \geq 3 h r(N)$).
    
    Combining~\eqref{eq:boundGInsteadOfSpan} and~\eqref{eq:boundCInS} we obtain
    \begin{equation}\label{eq:spanToChernoff}
        \Pr\left[|\mspan(D(S,h)\setminus S)|\leq \tfrac{\card{N}}{12}\right]
            \leq \Pr\left[\card{S} - \tfrac{\card{N}}{3} \leq \tfrac{\card{N}}{12}\right]
            \leq \Pr\left[\card{S} \leq \tfrac{5}{6} \E[\card{S}]\right],
    \end{equation}
    where the second inequality follows from $\E[\card{S}]=\sfrac{\card{N}}{2}$. Relation \eqref{eq:concentrationCard} now follows by applying a Chernoff bound $\Pr[X \leq (1 - \delta)\E[X]] < \mathtt{exp}[\sfrac{-\delta^{2} \E[X]}{2}]$ for $X = \card{S}$ to the right-hand side expression in \eqref{eq:spanToChernoff} and using $\E[\card{S}] = \sfrac{\card{N}}{2}$.
    
	We next prove the concentration result~\eqref{eq:concentrationRankS}.
	Let $B$ be a union of $3h$ disjoint bases contained in $\M$. We first show that for any set $A \subseteq B$ the following inequality holds:
	\begin{equation}\label{eq:rankConcentrationClaim}
		r(D(A, h)) \geq \frac{\card{A} - h r(A)}{2h}.
	\end{equation}
	To this end, we start by observing that, because $\mrestrict{\M}{B}$ is a uniformly dense matroid with density $3h$, we have
	\begin{equation}\label{eq:noDenserInUnifDense}
	3h r(Q)  \geq |Q| \qquad \forall Q\subseteq B.
	\end{equation}
	The above property also immediately follows by observing that, if we write $B=B_1 \cup \cdots \cup B_{3h}$ as the union of $3h$ disjoint bases, then we have
	\begin{equation*}
	|Q| = \sum_{i=1}^{3h} |Q \cap B_i| = \sum_{i=1}^{3h} r(Q\cap B_i) \leq \sum_{i=1}^{3h} r(Q) = 3h r(Q).
	\end{equation*}
	Relation~\eqref{eq:rankConcentrationClaim} now holds due to the following:
	\begin{align*}
	2h r(D(A,h)) \geq |D(A,h)| - h r(D(A,h)) \geq |A| - h r(A),
	\end{align*}
	where the first inequality is due to~\eqref{eq:noDenserInUnifDense} and the second one follows from the definition of $D(A,h))$.

	Now, in order to show \eqref{eq:concentrationRankS}, suppose that the event $r(D(S, h)) \leq \sfrac{r(N)}{8}$ occurs. Then
	\[
		\frac{r(N)}{8}
		\geq r(D(S, h))
		\geq r(D(S \cap B, h))
		\geq \frac{\card{S \cap B} - h r(S \cap B)}{2h}
		\geq \frac{\card{S \cap B} - h r(N)}{2h},
	\]
	where the third inequality holds by \eqref{eq:rankConcentrationClaim}. Since rearranging the terms in the above expression gives $\card{S \cap B} \leq \sfrac{5h r(N)}{4}$, we have $\Pr\left[r(D(S, h)) \leq \sfrac{r(N)}{8}\right] \leq \Pr\left[\card{S \cap B} \leq \sfrac{5h r(N)}{4}\right]$. To upper bound the latter probability, observe that $S \cap B$ contains each element of $B$ with probability $\sfrac{1}{2}$ independently by construction of $S$. Therefore we can apply a Chernoff bound $\Pr\left[X \leq (1 - \delta) \E[X]\right] \leq \exp\left(-\sfrac{\delta^{2}\E[X]}{2}\right)$ with $X = \card{S \cap B}$, $\E[X] = \sfrac{|B|}{2} = \sfrac{3h r(N)}{2}$, and $\delta = \sfrac{1}{6}$ (chosen so that $(1 - \delta) \E[X] = \sfrac{5h r(N)}{4}$), resulting in
	\begin{align*}
		\Pr\left[r(D(S, h)) \leq \frac{r(N)}{8}\right]
		& \leq  \mathtt{exp} \left(-\frac{h r(N)}{48} \right)
		\leq  \mathtt{exp} \left(-\frac{r(N)}{48}\right),
	\end{align*}
	where the last inequality holds since $h \geq 1$.
\end{proof}

The proof of \cref{thm:goodEvent} is based on the concentration result \eqref{eq:concentrationRankS}. In summary, rather than directly showing that $\rho_{\mrestrict{\M}{S}}$ approximates $\rho_{\M}$ well everywhere, we consider a discrete set of points on $\rho_{\M}$ associated to minors of $\M$ of geometrically increasing ranks. We then apply \eqref{eq:concentrationRankS} to these minors and employ a union bound to show that we get a good approximation for these grid points. The union bound works out because the ranks are geometrically increasing and appear in the exponent of the right-hand side of \eqref{eq:concentrationRankS}. The complete proof is presented below.

\begin{proof}[Proof of \cref{thm:goodEvent}]
	Let $\{\lambda_i\}_{i \in [m]}$ denote the densities (i.e., values of at least one) in the image of $\rho_{\M}$, and let $\tau \coloneqq \max\{t>0 \colon \rho_{\M}(t) \geq 1\}$.	
	Let $\rho'\colon \mathbb{R}_{>0} \to \mathbb{R}_{\geq 0}$ be the curve obtained from $\rho_{\M}$ by rounding down every density $\lambda_i$ to the closest power of $3$.
	That is,
	\begin{equation*}
		\rho'(t) =
		\begin{cases}
		\max\{3^j \colon 3^j \leq \rho_{\M}(t), j \in \Z_{\geq 0}\} &\text{if $0 < t \leq \tau$},\\
		0 &\text{if $t > \tau$}.
		\end{cases}
	\end{equation*}
	Let $\lambda'_1 > \cdots >\lambda'_{\ell}$ denote the densities in the image of $\rho'$, and note that by construction all the $\lambda'_i$ are powers of $3$. In particular, $\rho'$ is a $(1,3)$-approximation of $\rho_{\M}$.
	
	Next, let $r_{\max} \colon \R_{\geq 0} \to \R_{\geq 0}$ be the function given by
	\begin{equation*}
		r_{\max}(\lambda) =
		\begin{cases}
			\max\{t: \rho'(t) \geq \lambda\} &\text{if $0 \leq \lambda \leq \lambda'_1$},\\
			0 &\text{if $\lambda > \lambda_1'$}.
		\end{cases}
	\end{equation*}
	We define a subset of densities of $\{\lambda'_i\}_{i \in [\ell]}$ as follows. First, set $\mu_1=\rho'(36)$ and $i=1$. Then, while $\rho'(36 r_{\max}(\mu_i)) \geq 1$ (i.e., $\rho'(36 r_{\max}(\mu_i))$ is in the non-zero support of $\rho'$), set $\mu_{i+1} = \rho'(36 r_{\max}(\mu_i))$ and update $i=i+1$.
	Let $\Lambda$ denote the subset of densities selected by the above procedure and let $q\coloneqq |\Lambda|$.
	
	Now, for each $i \in [q]$, define $r_i \coloneqq r(D(N,\mu_i))$. Note that $r_i = r_{\max}(\mu_i)$. We then call a subset $S \subseteq N$ \emph{good} if it satisfies
	\begin{equation}\label{eq:goodApprox}
		r\left(D\left(S, \frac{\mu_i}{3}\right)\right) \geq \frac{r_i}{8} \qquad \forall i \in [q].
	\end{equation}
	The motivation for the above definition is that any good set $S$ satisfies that $\rho_{\mrestrict{\M}{S}}$ is a $(288,9)$-approximation of $\rho_{\M}$. To see this, first note that if $S$ is good then $\rho_{\mrestrict{\M}{S}} (t) \geq \sfrac{\mu_i}{3}$ for all $t\in (0,\sfrac{r_i}{8}]$. 
	Next, let $\mu_{i-1},\mu_i \in \Lambda$ and $t\in (\sfrac{r_{i-1}}{8},\sfrac{r_{i}}{8}]$. Then $288t > 36r_{i-1}$, and thus $\rho'(288t) \leq \rho'(36r_{i-1})=\mu_i$, where the last equality follows by construction of the $\mu_i$'s. Using that $\rho'$ is a $(1,3)$-approximation of $\rho_{\M}$, it follows that $\rho_{\M}(288t)\leq 3 \rho'(288t) \leq 3 \mu_i$. Combining all the above, it follows that 
	\begin{equation*}
		\rho_{\mrestrict{\M}{S}} (t) \geq \frac{\mu_i}{3} \geq \frac{\rho_{\M}(288t)}{9} \qquad \forall \mu_{i-1},\mu_i \in \Lambda \text{ and } t\in \left(\frac{r_{i-1}}{8},\frac{r_{i}}{8}\right].
	\end{equation*}
	Moreover, notice that for $t\in (1,\sfrac{r_1}{8}]$ we have $\rho'(288t) \leq \rho'(36) = \mu_1$. Hence, using the same reasoning as above, it follows that $\rho_{\mrestrict{\M}{S}} (t) \geq \sfrac{\mu_1}{3} \geq \sfrac{\rho_{\M}(288t)}{9}$.
	Finally, consider the case where $t > \sfrac{r_q}{8}$. 
	By construction of the $\mu_i$ we have $\rho'(288t) \leq \rho'(36r_q) = 0$. 
	Hence $\rho_{\M}(288t)=\rho'(288t)=0$, since $\rho'$ is a $(1,3)$-approximation of $\rho_{\M}$ and thus $\rho'(a)=0 \iff \rho_{\M}(a)=0$ for any $a>0$.
	It follows that $\rho_{\mrestrict{\M}{S}} (t) \geq 0 = \rho_{\M}(288t)$.
	This concludes the proof of the claim that if $S$ is good, then $\rho_{\mrestrict{\M}{S}}$ is a $(288,9)$-approximation of $\rho_{\M}$.	
		
	Hence, in order to prove the theorem it remains to show that the probability of a set $S\subseteq N$ being good (i.e., satisfying \eqref{eq:goodApprox}) is at least $\sfrac{1}{100}$. We discuss this next. First, note that by \cref{eq:concentrationRankS} from \cref{thm:concentrationBounds}, for each $\mu_i \in \Lambda$ with $\mu_i \geq 3$ it holds
	\begin{equation*}
		\Pr\left[r\left(D\left(S, \frac{\mu_i}{3}\right)\right) \leq \frac{r_i}{8}\right]
		\leq \mathtt{exp}\left(-\frac{r_i}{48}\right).
	\end{equation*}
	In the case where $\mu_q=1$, while we cannot directly black box the concentration bound from \cref{thm:concentrationBounds} since the assumptions are not met, we can still get the same bound as follows. Note that if $\mu_q=1$, then $r(D(S, \sfrac{\mu_q}{3}))$ is just $r(S)$, and $r_q = r(N)$.
	Let $B$ be any basis of $\M$, and observe that $r(S) \geq |S \cap B|$. Since $S$ contains every element of $N$ independently with probability $\sfrac{1}{2}$, we can use a Chernoff bound $\Pr\left[X \leq (1 - \delta) \E[X]\right] \leq \exp\left(-\sfrac{\delta^{2}\E[X]}{2}\right)$ with $X = \card{S \cap B}$, $\E[X] = \sfrac{|B|}{2} = \sfrac{r(N)}{2}$, and $\delta = \sfrac{3}{4}$ (chosen so that $(1 - \delta) \E[X] = \sfrac{r(N)}{8}$), resulting in
	\begin{equation*}
		\Pr\left[r\left(D\left(S, \frac{\mu_q}{3}\right)\right) \leq \frac{r_q}{8}\right]
		= \Pr\left[r(S) \leq \frac{r(N)}{8}\right]
		\leq \mathtt{exp}\left(-\frac{3^2 \cdot r_q}{4^3}\right)
		\leq \mathtt{exp}\left(-\frac{r_q}{48}\right).
	\end{equation*}

	Finally, note that by construction we have $r_1\geq 36$ and $r_{i+1} \geq 36 r_i$ for all $i,i+1 \in [q]$. 
	Thus, $r_i\geq 36^i$ for each $i \in [q]$, and hence by the union bound it follows
	\begin{equation*}
		\P[S \text{ not good}] 
		\leq \sum_{i=1}^q \mathtt{exp}\left(-\frac{r_i}{48}\right)
		\leq \sum_{i=1}^q \mathtt{exp}\left(-\frac{36^i}{48}\right)
		\leq \left[\mathtt{exp}\left(-\frac{36}{48}\right)+2\mathtt{exp}\left(-\frac{36^2}{48}\right)\right]
		\leq \frac{99}{200}.
	\end{equation*}
	Hence,
	\begin{equation*}
		\P[S \text{ and } N\setminus S \text{ are both good}] \geq 1 - 2 \cdot \frac{99}{200} = \frac{1}{100},
	\end{equation*}
	which concludes the proof.
\end{proof}
\section{The Main Algorithm and Its Analysis}\label{sec:algorithm}

In this section we describe and analyse the procedure from \cref{thm:aidedAlgorithm}. The analysis consists of two main ingredients. The first one is to show that if the approximate curve $\rhoapprox$ is well-struc\-tured (in some well-defined sense), then there is an algorithm retrieving a constant factor of $F(\rhoapprox)$ on expectation --- see \cref{thm:GRP}. The second one is then to show that, given any initial approximate curve $\rhoapprox$, one can find well-struc\-tured curves whose $F$ function value is close to $F(\rhoapprox)$ --- see \cref{thm:findGoodCurve}. 

The next result, proved in \cref{sec:proofThmGRP}, formalizes the first step above.

\begin{theorem}\label{thm:GRP}
	Let $\M=(N,\I)$ be a matroid with $w$-sampled weights, and let $r$ and $\rho_{\M}$ denote the rank function and rank-density curve of $\M$, respectively.
	Let $\rhogrid \leq \rho_{\M}$ be a rank-density curve with densities $\{\lambdagrid_{i}\}_{i \in [m]}$ such that the $\lambdagrid_{i}$ are powers of some integer $\beta \geq 3$ and $\lambdagrid_{1} > \cdots > \lambdagrid_{m} \geq 1$, and such that $r(D(N, \lambdagrid_{i + 1})) \geq 24 r(D(N, \sfrac{\lambdagrid_{i}}{\beta}))$ for $i \in [m - 1]$.
	Then there is an efficient procedure $\mathrm{ALG}(\rhogrid, \beta)$ that, when run on the RA-MSP subinstance given by $\M$, returns an independent set $I$ of $\M$ of expected weight at least $(\sfrac{1}{180e}) F(\rhogrid)$.
\end{theorem}

We note that the above result assumes the $\lambdagrid_{i}$ to be powers of $\beta$ mainly for convenience (so that $\sfrac{\lambdagrid_{i}}{\beta}$ is an integer), but it is not strictly needed.

The second main ingredient in the proof of \cref{thm:aidedAlgorithm} is the following result. 

\begin{theorem}\label{thm:findGoodCurve}
	Let $\M=(N,\I)$ be a matroid with $w$-sampled weights, and let $r$ and $\rho_{\M}$ denote the rank function and rank-density curve of $\M$, respectively.
	Given an $(\alpha, \beta)$-approximate curve $\rhoapprox$ of $\rho_{\M}$ with $\alpha \in \mathbb{R}_{\geq 24}$ and $\beta \in \mathbb{Z}_{\geq 3}$, there is a procedure $\mathrm{ALG}(\rhoapprox, \alpha, \beta)$ returning rank-density curves $\rhogrid, \rhogrid_{1}, \rhogrid_{2}, \rhogrid_{3}, \rhogrid_{4}$ such that:
	\begin{enumerate}
	    \item $\rhogrid$ is an $(\alpha^2, \beta^2)$-approximation of~$\rho_{\M}$. \label{item:goodCurve1}
	    \item $\sum_{i \in [4]} F(\rhogrid_{i}) \geq F(\rhogrid)$. \label{item:goodCurve2}
	    \item For each $i\in [4]$, $\rhogrid_{i}$ satisfies the following properties: Let $\{\mu_j\}_{j \in [\ell]}$ be the densities of $\rhogrid_{i}$, then all the $\mu_j$ are powers of $\beta \geq 3$, and $r(D(N, \mu_{j + 1})) \geq \alpha r(D(N, \sfrac{\mu_{j}}{\beta})) \geq 24 r(D(N, \sfrac{\mu_{j}}{\beta}))$ for $j \in [\ell - 1]$. Moreover, $\rhogrid_{i} \leq \rho_{\M}$. \label{item:goodCurve3}
	\end{enumerate}
\end{theorem}
\begin{proof}
	We first discuss how to build the curve $\rhogrid$.
	The goal is, on the one hand, to guarantee that the image $\{\lambdagrid_i\}_{i \in [\overline{m}]}$ of the curve $\rhogrid$ satisfies that each $\lambdagrid_i$ is a power of $\beta$, and moreover, that the ranks corresponding to any two consecutive densities $\lambdagrid_i$ and $\lambdagrid_{i+1}$ are at least an $\alpha$ factor apart (we formalize this below). 
	On the other hand, we want $\rhogrid$ to be as close as possible to $\rhoapprox$, while satisfying $\rhogrid \leq \rhoapprox$. 
	We next discuss how to achieve these.
	
	Let $\{\tilde{\lambda}_i\}_{i \in [\tilde{m}]}$ denote the densities (i.e., values of at least one) in the image of $\rhoapprox$, and let $\tau \coloneqq \max\{t>0 \colon \rhoapprox(t) \geq 1\}$.	
	Let $\rho'\colon \mathbb{R}_{>0} \to \mathbb{R}_{\geq 0}$ be the curve obtained from $\rhoapprox$ by rounding down every density $\tilde{\lambda}_i$ to the closest power of $\beta$.
	That is,
\begin{equation*}
\rho'(t) \coloneqq
\begin{cases}
\max\{\beta^j \colon \beta^j \leq \rhoapprox(t), j \in \Z_{\geq 0}\} &\text{if $0 < t \leq \tau$},\\
0 &\text{if $t > \tau$}.
\end{cases}
\end{equation*}
	Let $\{\lambda'_i\}_{i \in [m']}$ denote the densities in the image of $\rho'$, and note that by construction all the $\lambda'_i$ are powers of $\beta$. 
	It thus remains to guarantee that the geometric rank increase property for any two consecutive densities is satisfied. 	
	In order to achieve this, we use the following definition.
	Given a rank-density curve $\rho$ with image $\lambda_1 > \cdots > \lambda_m$, we define the function $r^{\rho}_{\max}:\R_{\geq 0} \to \R_{\geq 0}$ by
\begin{equation*}
r^{\rho}_{\max}(\lambda) \coloneqq
\begin{cases}
\max\{t: \rho(t) \geq \lambda\} &\text{if $\lambda \leq \lambda_1$},\\
0 &\text{otherwise}. 
\end{cases}
\end{equation*}
For brevity, let $r'_{\max}\coloneqq r^{\rho'}_{\max}$ denote the $r_{\max}$ function corresponding to the curve $\rho'$.
	
	We then build $\rhogrid$ as follows.		
	First, set $\lambdagrid_1 = \lambda'_1$ and $i=1$. 
	Then, while $\rho'(\alpha \cdot r'_{\max} (\lambdagrid_i)) \geq 1$, set $\lambdagrid_{i+1} = \rho'(\alpha \cdot r'_{\max} (\lambdagrid_i))$ and update $i=i+1$. 
	Let $\Lambdagrid \coloneqq \{\lambdagrid_i\}_{i \in [\overline{m}]}$ denote the densities selected by the above procedure. 
	We define $\rhogrid$ to be the curve obtained from $\rho'$ by further rounding down densities $\lambda'_i$ in the image of $\rho'$ to the closest $\lambdagrid \in \Lambdagrid$.
	More precisely,
\begin{equation*}
\rhogrid(t) = 
\begin{cases}
\max\{ \lambdagrid \in \Lambdagrid: \lambdagrid \leq \rho'(t)\} &\text{if $0 < t \leq \overline{\tau}$},\\
0 &\text{if $t > \overline{\tau}$},
\end{cases}
\end{equation*}
where $\overline{\tau} \coloneqq \overline{r}_{\max}(\lambdagrid_{\overline{m}}) \leq \tau$. 
	
	We have that $\rhogrid$ is an $(\alpha,\beta)$-approximation of $\rhoapprox$ by construction due to \Cref{lem:approxOfApprox}, because $\rho'$ is a $(1,\beta)$-approximation of $\rhoapprox$, and $\rhogrid$ is an $(\alpha,1)$-approximation of $\rho'$. 
	Combining this with the fact that $\rhoapprox$ is an $(\alpha,\beta)$-approximation of $\rho_{\M}$ proves property~\ref{item:goodCurve1}.
	
	Next, let $\tilde{r}_{\max}\coloneqq r^{\tilde{\rho}}_{\max}$ and $\overline{r}_{\max}\coloneqq r_{\max}^{\rhogrid}$ denote the $r_{\max}$ functions corresponding to the curves $\rhoapprox$ and $\rhogrid$, respectively. 
We show the following.
\begin{claim}\label{claim:geomRankProp}
If $|\Lambdagrid|\geq 5$, then for any five consecutive densities $\lambdagrid_{i}, \ldots, \lambdagrid_{i+4} \in \Lambdagrid$, we have 
\begin{equation}
\label{eq:geomRankPropIntermed}
r\left( D\left(N,\frac{\lambdagrid_i}{\beta}\right) \right)
	\leq \alpha \cdot r'_{\max} \left(\lambdagrid_{i+2}\right).
\end{equation}
\end{claim}
Let $\kappa \coloneqq r(D(N,\sfrac{\lambdagrid_i}{\beta}))$.
Note that if $\kappa \leq \alpha$, then \Cref{claim:geomRankProp} holds because $r'_{\max}(\lambdagrid_{i+2})\geq 1$ by construction.
Now assume $\kappa > \alpha$.
We claim that in this case we have
\begin{equation}\label{eq:relLargeKappa}
		\kappa \leq \alpha \cdot \tilde{r}_{\max} \left(\frac{\lambdagrid_i}{\beta^2}\right),
\end{equation}
which holds due to the following.
Because $\rhoapprox$ is an $(\alpha,\beta)$-approximation of $\rho$, we have
\begin{equation*}
\rhoapprox(t) \geq \frac{\rho(\alpha t)}{\beta} \qquad \forall t\geq 1.
\end{equation*}
Applying this inequality with $t=\sfrac{\kappa}{\alpha}$, we get
\begin{equation}\label{eq:relLargeKappaRho}
\rhoapprox\left(\frac{\kappa}{\alpha}\right) \geq \frac{\rho(\kappa)}{\beta} \geq \frac{\lambdagrid_i}{\beta^2},
\end{equation}
where the second inequality follows from $\kappa \coloneqq r(D(\sfrac{N,\lambdagrid_i}{\beta}))$.
Finally, \Cref{eq:relLargeKappa} is an immediate implication of \Cref{eq:relLargeKappaRho}.
The claim now follows due to
\begin{equation*}
\kappa 
		\leq \alpha \cdot \tilde{r}_{\max} \left(\frac{\lambdagrid_i}{\beta^2}\right)
		=\alpha \cdot r'_{\max} \left(\frac{\lambdagrid_i}{\beta^2}\right)
		\leq \alpha \cdot r'_{\max} \left(\lambdagrid_{i+2}\right),
\end{equation*}
where the first inequality is due to~\eqref{eq:relLargeKappa}, the first equality follows by construction of $\rho'$ and the fact that $\sfrac{\lambdagrid_i}{\beta^2}$ is a power of $\beta$ (because $\lambdagrid_i$ is a power of $\beta$ and $\lambdagrid_i \geq \beta^4 \lambdagrid_{i+4} \geq \beta^4$), and the second inequality again follows since by construction $\sfrac{\lambdagrid_i}{\beta^2} \geq \lambdagrid_{i+2}$.

\medskip

Using \cref{claim:geomRankProp}, we can further upper bound $r(D(N,\frac{\lambdagrid_i}{\beta}))$ as follows
\begin{equation}
		\label{eq:geomRankProp}
		r\left( D\left(N,\frac{\lambdagrid_i}{\beta}\right) \right)
		\leq \alpha \cdot r'_{\max} \left(\lambdagrid_{i+2}\right) 
		 = \alpha \cdot \overline{r}_{\max} (\lambdagrid_{i+2}) 
		\leq \frac{1}{\alpha} \cdot \overline{r}_{\max} (\lambdagrid_{i+4})
		\leq \frac{1}{\alpha} \cdot r( D(N,\lambdagrid_{i+4}) ),
	\end{equation}
where the first inequality is due to \Cref{claim:geomRankProp}, the equality holds because $r'_{\max} (\lambdagrid) = \overline{r}_{\max} (\lambdagrid)$ for every $\lambdagrid \in \Lambdagrid$ by construction, the second inequality holds because by construction $\overline{r}_{\max} (\lambdagrid_{j+1}) \geq \alpha \cdot \overline{r}_{\max} (\lambdagrid_{j})$ for any two consecutive densities $\lambdagrid_{j}, \lambdagrid_{j+1} \in \Lambdagrid$, and the last inequality follows from $\rhogrid \leq \rho$.

We now build the curves $\{\rhogrid_i\}_{i \in [4]}$ as follows.
	For $i\in [4]$, let
\begin{equation}
\Lambdagrid_i \coloneqq \left\{\lambdagrid_j \in \Lambdagrid: j \equiv i-1 \pmod 4 \right\}.
\end{equation}
	If $\Lambdagrid_i = \emptyset$, let $\rhogrid_i$ be (by convention) the curve given by $\rhogrid_i (t) = 1$ for $0 < t \leq \overline{\tau}$, and $\rhogrid_i (t)=0$ for $t > \overline{\tau}$.
	Otherwise, we define $\rhogrid_i$ to be the curve obtained from $\rhogrid$ by further rounding down densities $\lambdagrid \in \Lambdagrid$ to the closest density in $\Lambdagrid_i$.
	More precisely,
\begin{equation*}
\rhogrid_i (t) \coloneqq
\begin{cases}
\max\{ \mu \in \Lambdagrid_i: \mu \leq \rhogrid (t)\} &\text{if $0 < t \leq \overline{\tau}_i$},\\
 0 &\text{if $t > \overline{\tau}_i$},
\end{cases}
\end{equation*}
where $\overline{\tau}_i \coloneqq \overline{r}_{\max} (\min\{\mu: \mu \in \Lambdagrid_i\})$.
	
	To see property~\ref{item:goodCurve2}, note that by construction we have $\cup_{i \in [4]} \Lambdagrid_{i} = \Lambdagrid$. By setting $\overline{r}_{\max} (\lambdagrid_{i}) = 0$ for all $i \leq 0$ by convention, we get
	\begin{align*}
	    \sum_{i = 1}^{4} F(\rhogrid_{i})
              =  \sum_{i = 1}^{\overline{m}} (\overline{r}_{\max}(\lambdagrid_{i}) - \overline{r}_{\max}(\lambdagrid_{i - 4})) \eta(\lambdagrid_{i})
            \geq \sum_{i = 1}^{\overline{m}} (\overline{r}_{\max}(\lambdagrid_{i}) - \overline{r}_{\max}(\lambdagrid_{i - 1})) \eta(\lambdagrid_{i})
              =  F(\rhogrid).
	\end{align*}
 
	Finally, to see property~\ref{item:goodCurve3}, first note that by construction we have $\rhogrid_i \leq \rhogrid \leq \rho' \leq  \rho_{\M}$ for each $i \in [4]$, and moreover, all the curves $\rhogrid_i$ consist of densities that are powers of $\beta$. 
	In addition, the geometric rank increase property holds trivially if $|\Lambdagrid| \leq 4$. Otherwise, it follows directly from \cref{eq:geomRankProp} and the fact that $\alpha \geq 24$.
\end{proof}

We now show how \cref{thm:GRP} and \cref{thm:findGoodCurve} combined imply \cref{thm:aidedAlgorithm}.

\begin{proof}[Proof of \cref{thm:aidedAlgorithm}]
	Given an $(\alpha, \beta)$-approximation $\rhoapprox$ of $\rho_{\M}$, first run the procedure from \cref{thm:findGoodCurve} to get curves $\rhogrid,\rhogrid_1,\rhogrid_2,\rhogrid_3,\rhogrid_4$. Then choose an index $i\in [4]$ uniformly at random and run the procedure from \cref{thm:GRP} on $\rhogrid_i$ to get an independent set with expected weight at least
	\begin{equation*}
		\frac{1}{180e}  \left( \frac{1}{4} \sum_{i=1}^4 F(\rhogrid_{i}) \right)
		    \geq \frac{1}{720e} F(\rhogrid)
		    \geq \frac{1}{1440 e \alpha^2 \beta^2} \left(F(\rho_{\M}) - \alpha^2 w_{\max} \right),
	\end{equation*}
	where the last inequality uses \cref{lem:shiftedCurve} and the fact that $\rhogrid$ is an $(\alpha^{2}, \beta^{2})$-approximation of $\rho_{\M}$.
\end{proof}

Thus, to show \cref{thm:aidedAlgorithm}, it remains to prove \cref{thm:GRP}.

\subsection{Proof of \cref{thm:GRP}}\label{sec:proofThmGRP}

Throughout this section we use the notation and assumptions from \cref{thm:GRP}.

We prove the theorem in two steps. First, we show that after observing a sample set $S$, we can build a chain $\bigoplus_{i = 1}^{m} \M_{i}$ of $\mrestrict{\M}{N \setminus S}$ satisfying certain properties with at least constant probability. Then we show that, given such a chain, there is a procedure returning an independent set $I$ of $\M$ with $\E[w(I)] = \Omega(F(\rhogrid))$, leading to the desired result.
We start by discussing the former claim. 

Given a sample set $S \subseteq N$, we build a chain of matroids as follows. For $i \in [m]$ let
\begin{equation}\label{eq:chain}
\begin{aligned}
	N_{i} &\coloneqq \mspan\left(D\left(S, \frac{\lambdagrid_{i}}{\beta}\right)\right) \setminus \left(S \cup \mspan\left(D\left(S, \frac{\lambdagrid_{i - 1}}{\beta}\right)\right)\right) \mbox{, and }\\
\M_{i} &\coloneqq \mminor{\M}{\mspan\left(D\left(S, \frac{\lambdagrid_{i - 1}}{\beta}\right)\right)}{N_{i}},
\end{aligned}
\end{equation}
where $D(S, \sfrac{\lambdagrid_{0}}{\beta}) = \emptyset$ by convention.

In addition, for every $i \in [m]$ let $\overline{N}_{i} \coloneqq D(N, \lambdagrid_{i})$, and define $\Lambdagrid \coloneqq \{i \in [m] \colon r(\overline{N}_{i}) \geq 24, \ \lambdagrid_{i} \geq \beta\}$. 
Note that $\Lambdagrid$ and the $\overline{N}_{i}$'s do not depend on the sample set $S$.
Moreover, from the assumptions of \cref{thm:GRP} it follows that $\Lambdagrid \supseteq [m] \setminus \{1,m\}$.
The next result shows that with constant probability, the sample set $S$ is such that for each $i \in \Lambdagrid$, the set $N_i$ contains a subset $U_i$ of large rank and density; more precisely, $r(U_i) \geq \Omega(r(\overline{N}_i))$ and $\sfrac{|U_i|}{r(U_i)} \geq \Omega(\overline{\lambda}_i)$.
\begin{lemma}\label{lem:goodSample}
	Let $S \sim \BinomDistr(N, \sfrac{1}{2})$, and let $N_i$, $\overline{N}_i$, and $\Lambdagrid$ be as defined above. Then, with probability at least $\sfrac{1}{3}$, every $N_{i}$ with $i \in \Lambdagrid$ contains $\lambdagrid_{i}$ disjoint independent sets $I_{1}, \ldots, I_{\lambdagrid_{i}}$ such that $\sum_{j=1}^{\lambdagrid_{i}} \card{I_{j}} \geq (\sfrac{1}{24}) \lambdagrid_{i} r(\overline{N}_{i})$.
\end{lemma}
\begin{proof}
	Let $B_{i} \subseteq \overline{N}_{i}$ denote the union of (any) $\lambdagrid_{i}$ disjoint bases of $\mrestrict{\M}{\overline{N}_{i}}$.
	We then say that a sample set $S$ is \emph{good} if it satisfies
	\begin{equation}\label{eq:goodS}
		\left|\left(\mspan\left(D\left(S \cap B_{i}, \frac{\lambdagrid_{i}}{\beta}\right)\right) \cap B_{i}\right) \setminus S\right| \geq \frac{1}{12} \cdot \card{B_{i}} \qquad \forall i \in \Lambdagrid.
	\end{equation}
	The motivation for the above definition is that any good sample set $S$ leads to a matroid chain (as defined in \cref{eq:chain}) that satisfies the properties claimed in the lemma. 
	To see this, note that for every $i \in \Lambdagrid \cap [m - 1]$ we have
	\begin{align*}
		\left\vert\mspan\left(D\left(S, \frac{\lambdagrid_{i}}{\beta}\right)\right) \cap B_{i + 1}\right|
		&\leq \left\vert\mspan\left(D \left(N, \frac{\lambdagrid_{i}}{\beta}\right)\right) \cap B_{i + 1}\right\vert\\
		&\leq \lambdagrid_{i + 1} r\left(D\left(N, \frac{\lambdagrid_{i}}{\beta}\right)\right)\\
		&\leq \frac{\lambdagrid_{i + 1}}{24} r(\overline{N}_{i + 1}),
	\end{align*}
	where the last inequality holds since $r(D(N, \lambdagrid_{i + 1})) \geq 24 r(D(N, \sfrac{\lambdagrid_{i}}{\beta}))$ by the assumptions of \cref{thm:GRP}. Moreover, since  $D(S, \sfrac{\lambdagrid_{0}}{\beta}) = \emptyset$ by convention, the same bound holds for $i = 0$, i.e., $\card{\mspan\left(D\left(S, \sfrac{\lambdagrid_{0}}{\beta}\right)\right) \cap B_{1}} \leq (\sfrac{\lambdagrid_{1}}{24}) r(\overline{N}_{1})$. Thus, for any $i \in \Lambdagrid$ we get
	\begin{align*}
		\card{N_{i} \cap B_{i}}
		& \geq \left\vert\left(\mspan\left(D\left(S \cap B_{i}, \frac{\lambdagrid_{i}}{\beta}\right)\right) \cap B_{i}\right) \setminus S\right\vert - \left\vert\mspan\left(D\left(S, \frac{\lambdagrid_{i - 1}}{\beta}\right)\right) \cap B_{i}\right\vert \\
		&\geq \frac{1}{12} \card{B_{i}} - \frac{1}{24} \lambdagrid_{i} r(\overline{N}_{i})\\
		&= \frac{1}{24} \lambdagrid_{i} r(\overline{N}_{i}).
	\end{align*}
	Finally, since $B_{i}$ is a union of $\lambdagrid_{i}$ disjoint bases, then $N_{i} \cap B_{i}$ is a union of $\lambdagrid_{i}$ disjoint independent sets (contained in $N_{i}$), and hence the claim follows.
	
	Thus, in order to prove the lemma, it is only left to show that $\Pr\left[S \text{ is good}\right] \geq \sfrac{1}{3}$.
	To see this, first observe that since for all $i\in \Lambdagrid$ it holds that $\lambdagrid_{i} = \beta^j$ for some $j \geq 1$, and $\beta \geq 3$, then
	\begin{equation}\label{eq:auxGoodS}
		\frac{\lambdagrid_{i}}{\beta} 
    = \left\lfloor \frac{\lambdagrid_{i}}{\beta}\right\rfloor
    \leq \left\lfloor \frac{\lambdagrid_{i}}{3}\right\rfloor.
	\end{equation}

	Then, by applying \eqref{eq:concentrationCard} from \cref{thm:concentrationBounds} to every $\mrestrict{\M}{B_{i}}$ with $i \in \Lambdagrid$, we can bound the probability of $S$ not being good because of index $i\in \Lambdagrid$ by
	\begin{align*}
		\Pr&\left[\left\vert\left(\mspan\left(D\left(S \cap B_{i}, \frac{\lambdagrid_{i}}{\beta}\right)\right) \cap B_{i}\right) \setminus S\right\vert \leq \frac{1}{12} \card{B_{i}}\right] \\
		&\leq \Pr\left[\left\vert\left(\mspan\left(D\left(S \cap B_{i}, \left\lfloor \frac{\lambdagrid_{i}}{3} \right\rfloor\right)\right) \cap B_{i}\right) \setminus S\right\vert \leq \frac{1}{12} \card{B_{i}}\right]\\
		&\leq \exp\left(-\frac{\card{B_{i}}}{144}\right)\\
		&\leq \exp\left(-\frac{r(\overline{N}_{i})}{48}\right),
	\end{align*}
	where the first inequality follows from \cref{eq:auxGoodS}, the second from \cref{thm:concentrationBounds} and the fact that for all $i \in \Lambdagrid$ the matroid $\mrestrict{\M}{B_{i}}$ contains $3h$ disjoint bases with $h=\lfloor \sfrac{\lambdagrid_{i}}{3} \rfloor \geq 1$, and the last inequality holds since $\card{B_{i}} = \lambdagrid_{i} r(\overline{N}_{i})$ and $\lambdagrid_{i} \geq \beta \geq 3$ for all $i \in \Lambdagrid$. 
	
We can now upper bound the probability of $S$ not being good by using a union bound over $i\in \Lambdagrid$ and $r(\overline{N}_{i}) \geq 24$ for $i\in \Lambdagrid$ as follows:
	\begin{equation*}
		\sum_{i \in \Lambdagrid} \exp\left(-\frac{r(\overline{N}_{i})}{48}\right)
		\leq \sum_{i=1}^m \exp\left(-\frac{24^i}{48}\right)
		\leq \exp\left(-\frac{24}{48}\right) + 2 \exp\left(-\frac{24^2}{48}\right)
		\leq \frac{2}{3}.
	\end{equation*}
Hence, $\Pr\left[S \text{ is good}\right] \geq 1 - \sfrac{2}{3} = \sfrac{1}{3}$, as desired.
\end{proof}

The second main ingredient in the proof is to show that the above result can be exploited algorithmically. 
More precisely, we prove the following.

\begin{lemma}\label{lem:ospProcedure}
	Let $\M=(N,\I)$ be a matroid with $w$-sampled weights that contains $h$ disjoint independent sets $I_{1}, \dots, I_{h}$ such that $s \coloneqq (\sfrac{1}{h}) \sum_{j = 1}^{h} \card{I_{j}} \geq 1$. 
	Then there is a procedure that, when run on the RA-MSP subinstance given by $\M$, and with only $h$ given upfront, returns an independent set of $\M$ with expected weight at least $(\sfrac{s}{2e}) \eta(h)$. This is still the case even if the elements of $\M$ are revealed in adversarial (rather than uniformly random) order.
\end{lemma}
\begin{proof}
	Suppose we run the online selection procedure (OSP) described in \cref{alg:osp} on the RA-MSP subinstance given by $\M$, and with parameter $h$ (as defined in the statement of \cref{lem:ospProcedure}) as input.
	
	\begin{algorithm}[ht]
		\caption{Online selection procedure (OSP)\label{alg:osp}}
		
		\KwInput{A parameter $h \in \Z_{\geq 1}$, and an RA-MSP subinstance $\M=(N,\I)$ revealed online in any arbitrary order.}
		\KwOutput{Independent set $I \in \I$.}
		
		Set $I \leftarrow \emptyset$ \tcp*[r]{set of picked elements}
		
		\While{there are elements remaining in $\M$}{
			Initialize the classical MSP algorithm on $h$ elements\;
			
			Set $C \leftarrow \emptyset$ \tcp*[r]{input of the classical MSP algorithm}
			Set $\overline{e} \leftarrow \mathrm{None}$ \tcp*[r]{picked element, if any}
			
			\While{$\card{C} < h$ and there are elements remaining in $\M$}{
				Let $(e, w(e))$ be the next online input\;
				\If{$I \cup \{e\} \in \I$}{
					Feed $w(e)$ to the classical MSP algorithm and add $e$ to $C$\;
					\If{$w(e)$ is picked by the classical MSP algorithm}{Pick $e$ and set $\overline{e} \leftarrow e$}
				}
			}
			
			\lIf{$\overline{e}$ is not $\mathrm{None}$}{Add $\overline{e}$ to $I$}
		}
		Return $I$\;
	\end{algorithm}
	
	Suppose that OSP successfully completed its $i$-th iteration (where $i \in \{0, \dots, r - 1\}$) and is about to begin the $(i + 1)$-st iteration. Let $Z$ denote the set of elements seen so far and let $I$ be the set of elements picked so far. Additionally, let $T \coloneqq \bigcup_{j=1}^{h} I_{j}$ and $J \coloneqq \{e \in N \colon I \cup \{e\} \in \I\}$. First, note that $I \in \I$ by construction. Now, observe that
	\[
	\card{I_{j} \cap J} \geq \max \left\{0, \card{I_{j}} - \card{I}\right\} \geq \card{I_{j}} - i \qquad \forall j \in [h],
	\]
	since $I_{j}, I \in \I$ and $\card{I} \le i$. Therefore $\card{T \cap J} \geq (s - i)h$. Moreover, since the classical secretary algorithm was executed on $i$ groups of $h$ elements, we have $\card{(T \cap Z) \cap J} \le ih$ and hence
	\[
	\card{(N \setminus Z) \cap J} \geq \card{(T \setminus Z) \cap J} = \card{T \cap J} - \card{(T \cap Z) \cap J} \geq (s - 2i) h.
	\]
	Thus, if $s \geq 2i + 1$, the number $|(N\setminus Z) \cap J|$ of elements that \cref{alg:osp} can feed to the classical MSP algorithm in iteration $i+1$ is at least $h$, which implies that it successfully completes its $i + 1$ iteration; hence the total number of successful iterations is at least $\lfloor \sfrac{(s - 1)}{2} \rfloor + 1 \geq \max\{1, \, s / 2\}$.
	
	Finally, note that the expected weight of the element picked in each successful iteration is at least $\eta(h) / e$. Indeed, even though the arrival order is adversarial, due to the random assignment of the weights the classical secretary algorithm is applied to $h$ weights drawn uniformly at random from $w$. Combining this with the lower bound on the number of successful iterations gives the desired result.
\end{proof}

We can now combine \cref{lem:goodSample,lem:ospProcedure} to prove \cref{thm:GRP} as follows.

\begin{proof}[Proof of \cref{thm:GRP}.]
	Let $\mathrm{OSP}(\M,h)$ denote the online selection procedure from \cref{lem:ospProcedure} (i.e., \cref{alg:osp}). Additionally, for $i\in [m]$, let $r_{i}$ denote the coefficient of $\eta(\lambdagrid_{i})$ in $F(\rhogrid)$. Hence, $F(\rhogrid) = \sum_{i = 1}^{m} r_{i} \eta(\lambdagrid_{i})$. Consider the following algorithm: choose and execute one of the three branches presented below with probability $\sfrac{12}{15}$, $\sfrac{2}{15}$, and $\sfrac{1}{15}$, respectively.
	\begin{enumerate}
		\item\label{item:algGeneric} Observe $S \sim \BinomDistr(N, \sfrac{1}{2})$, construct the chain $\bigoplus_{i = 1}^{m} \M_{i}$ as defined in \eqref{eq:chain}, and run $\mathrm{OSP}(\M_{i}, \lambdagrid_{i})$ for every $i \in [m]$ (independently in parallel), returning all the picked elements.
		\item\label{item:algClassical} Run the classical secretary algorithm on $\M$ and return the picked element (if any).
		\item\label{item:algGreed} Run $\mathrm{OSP}(\M, 1)$ without observing anything and return all picked elements.
	\end{enumerate}
	
\noindent	Suppose we execute branch \ref{item:algGeneric}. By \cref{lem:goodSample}, with probability at least $\sfrac{1}{3}$, every $\M_{i}$ with $i \in \Lambdagrid$ satisfies the conditions of \cref{lem:ospProcedure} with parameters $h = \lambdagrid_{i}$ and $s = (\sfrac{1}{24}) r(\overline{N}_{i})$. Note that $s \geq 1$ holds given that $r(\overline{N}_{i}) \geq 24$ for all $i \in \Lambdagrid$. As additionally all matroids in the chain form a direct sum, executing the first branch of the algorithm returns an independent set with expected weight, due to \Cref{lem:ospProcedure}, of at least
	\[
	\frac{1}{3} \sum_{i \in \Lambdagrid} \frac{1}{2e} \cdot \frac{r(\overline{N}_{i})}{24} \eta(\lambdagrid_{i})
	= \frac{1}{144e} \sum_{i \in \Lambdagrid} r(\overline{N}_{i}) \eta(\lambdagrid_{i})
	\geq \frac{1}{144e} \sum_{i \in \Lambdagrid} r_i \eta(\lambdagrid_{i}),
	\]
	where the inequality follows from $\rhogrid \leq \rho_{\M}$ and $\overline{N}_{i} = D(N, \lambdagrid_{i})$ for every $i \in [m]$. 
	
	Therefore, if $i \in \Lambdagrid$, then the corresponding term $r_{i} \eta(\lambdagrid_{i})$ in $F(\rhogrid)$ is accounted for by branch \ref{item:algGeneric}. Thus it only remains to consider $i \in [m] \setminus \Lambdagrid \subseteq \{1,m\}$.
	
	Assume first that $1 \notin \Lambdagrid$. In this case, we must have $r(\overline{N}_1)<24$.  
	Since the expected weight yielded by running the classical secretary algorithm is at least $\sfrac{\eta(|N|)}{e}$, and $\eta(|N|) \geq \eta(\lambdagrid_{1})$, then, by running branch \ref{item:algClassical}, the expected weight of the output set is at least
	\begin{equation*}
		\frac{\eta(|N|)}{e} 
		\geq \frac{1}{e} \cdot \frac{r(\overline{N}_1) \eta(\lambdagrid_{1})}{r(\overline{N}_1)}
		\geq \frac{1}{23e} r_1 \eta(\lambdagrid_{1}),
	\end{equation*}
	where the last inequality follows from $r_1 \leq r(\overline{N}_1) \leq 23$.
	
	Finally, assume that $m \notin \Lambdagrid$. Then $\overline{\lambda}_{m} = 1$, in which case running branch \ref{item:algGreed} yields
	\begin{equation*}
		\E[w(\mathrm{OSP}(\M, 1))] \geq \frac{1}{2e} r(N) \eta(1) \geq \frac{1}{2e} r(\overline{N}_{m}) \eta(\lambdagrid_{m}),
	\end{equation*}
	where the first inequality holds by \cref{lem:ospProcedure} with $h = 1$ and $s = r(N) \geq 1$, as any basis of $\M$ is an independent set of rank $r(N)$, and the second inequality holds because $r(N) \geq r(\overline{N}_{m})$ and $\lambdagrid_{m} = 1$.
	
	The desired lower bound on the expected weight of the set returned by the algorithm now follows by combining the above results with the respective probabilities that each branch is executed.
\end{proof}

Finally, we discuss that our main result (i.e., \cref{thm:main}) still holds in the more general adversarial order with a sample setting, where we are allowed to sample a set $S \subseteq N$ containing every element of $N$ independently with probability $\sfrac{1}{2}$, and the remaining (non-sampled) elements arrive in adversarial order. 

In order to see this, first note that the only place in the proof of \cref{thm:GRP} where we use that the non-sampled elements (i.e., $N\setminus S$) arrive in random order, is to argue that when running the classical secretary algorithm in branch \ref{item:algClassical} we obtain an expected weight of at least $\sfrac{w_{\max}}{e}$. 
Indeed, branches \ref{item:algGeneric} and \ref{item:algGreed} rely on running the procedure from \cref{lem:ospProcedure}, whose guarantees hold in the case where the elements arrive in adversarial order.
However, note that running the classical secretary procedure in the above adversarial order with a sample setting outputs an element with expected weight of at least $\sfrac{w_{\max}}{4}$.
Indeed, the probability of selecting $w_{\max}$ in the latter setting is at least the probability of the event that $w_{\max}$ is not sampled and the second largest weight is; which occurs with probability $\sfrac{1}{4}$.
Thus, \cref{thm:GRP} holds (up to possibly a slightly worse constant) in the adversarial order with a sample setting.

Next, observe that this implies that \cref{thm:aidedAlgorithm} also holds in the above setting (again, up to possibly a slightly worse constant). This follows because its proof relies on combining the procedures from \cref{thm:GRP,thm:findGoodCurve}, and the latter is completely oblivious to the arrival order of the elements. 

Finally, note that the proof of \cref{thm:main} uses the procedure from \cref{thm:GRP} and the classical secretary algorithm. Because (as discussed above) both of these algorithms have very similar guarantees in the adversarial order with a sample setting to the ones shown in this paper for random order, the claim follows.

\subsection{Full Algorithm and Adversarial Order with a Sample Setting}

We next summarize the full algorithm used to obtain \cref{thm:main}.
First, the algorithm executes one of the following two branches uniformly at random:
\begin{enumerate}[label=(\theenumi)]
	\item \label{item:clasicalSec1} Run the classical secretary algorithm on $N$.
	\item \label{item:mainAlg} Run the following procedure:
\begin{itemize}
	\item Sample a set $S$ (without selecting anything) containing every element of $N$ with probability $\sfrac{1}{2}$ independently.
	\item Define $\tilde{\rho}$ to be the $(288,9)$-downshift of $\rho_{\mrestrict{\M}{S}}$.
	\item \label{item:thm35} Run the procedure of \cref{thm:aidedAlgorithm} on the remaining non-sampled elements (i.e., on the matroid $\mrestrict{\M}{N \setminus S}$) using as input the curve $\tilde{\rho}$, and parameters $\alpha = 288^2$ and $\beta = 9^2$.
\end{itemize}
\end{enumerate}
The procedure from \cref{thm:aidedAlgorithm} used above consists of two main steps: First, it runs the algorithm from \cref{thm:findGoodCurve} on the curve $\tilde{\rho}$ with parameters $\alpha = 288^2$ and $\beta = 9^2$, to find \emph{well-structured} curves $\overline{\rho}_1, \overline{\rho}_2, \overline{\rho}_3$, and $\overline{\rho}_4$. Then, it selects one $\overline{\rho}_j$ out of these four curves uniformly at random, and runs the procedure from \cref{thm:GRP} on the matroid $\mrestrict{\M}{N \setminus S}$ using as input such $\overline{\rho}_j$ and $\beta = 9^2$. The latter algorithm consists of executing one of the following three branches with probability $\sfrac{2}{15}$, $\sfrac{1}{15}$, and $\sfrac{12}{15}$ respectively:
\begin{enumerate}[label=(2.\roman*), leftmargin=3em]
	\item \label{item:clasicalSec2} Run the classical secretary algorithm on $N \setminus S$.
	\item \label{item:osp1} Run $\mathrm{OSP}(\mrestrict{\M}{N \setminus S}, 1)$.
	\item \label{item:osp2} Sample a set $S'$ (without selecting anything) containing every element of $N \setminus S$ with probability $\sfrac{1}{2}$ independently, construct the chain $\bigoplus_{i = 1}^{m} \M_{i}$ as defined in \eqref{eq:chain} using $\overline{\rho}_j$ and $\beta = 9^2$ as input, and run $\mathrm{OSP}(\M_{i}, \lambdagrid_{i})$ for every $i \in [m]$.
\end{enumerate}

The above completes the description of the full algorithm. 
In summary, the algorithm executes one of the following four options: branch \ref{item:clasicalSec1} with probability $\sfrac{1}{2}$, branch  \ref{item:clasicalSec2} with probability $\sfrac{1}{15}$, branch \ref{item:osp1} with probability $\sfrac{1}{30}$, and branch \ref{item:osp2} with probability $\sfrac{2}{5}$.
Hence, either the classical secretary algorithm is executed (on branch \ref{item:clasicalSec1} or \ref{item:clasicalSec2}), or the $\mathrm{OSP}$ procedure from \cref{alg:osp} is executed (on branch \ref{item:osp1} or \ref{item:osp2}). 
The guarantees of the latter hold even if the elements arrive in adversarial order --- see \cref{lem:ospProcedure}.
Moreover, because of the random assignment setting, the standard guarantees for the classical secretary algorithm still hold under adversarial arrival order.
Indeed, since weights are assigned uniformly at random to elements after the arrival order has been fixed, this setting is equivalent to the one where weights are assigned to elements adversarially but the arrival order is uniformly at random.

We now discuss how to adapt the above procedure and its analysis to the adversarial order with a sample setting.
In this setting, the algorithm must specify a (possibly random) sampling probability $p \in [0,1]$.
The instance is then revealed to the algorithm in two phases. 
First, a random set $S$ containing every element of $N$ with probability $p$ independently is provided to the algorithm.
However, the algorithm is not allowed to select any element from $S$.
In the second phase, the elements of $N \setminus S$ are then revealed in an adversarial order.

Our modified algorithm works as follows. 
First, it chooses one of the four options with the respective (i.e., $\sfrac{1}{2}$, $\sfrac{1}{15}$, $\sfrac{1}{30}$, and $\sfrac{2}{5}$) probabilities.
Then:
\begin{itemize}
	\item If branch \ref{item:clasicalSec1} is chosen, the algorithm sets $p=\sfrac{1}{e}$ for the sampling phase to obtain a sample set $S$.
	It then uses this sample set $S$ to choose a threshold, and picks the first element arriving in $N \setminus S$ with weight above the threshold (if any). 
	As discussed above, since for the classical secretary problem the random assignment adversarial order setting is equivalent to the standard adversarial assignment random order setting, this procedure outputs an element of expected weight at least $\sfrac{w_{\max}}{e}$.
	Thus being equivalent to running branch \ref{item:clasicalSec1} in the original algorithm.
	
	\item If branch \ref{item:clasicalSec2} is chosen, the algorithm sets $p=\sfrac{(e+1)}{2e}$ for the sampling phase to obtain a sample set $S$. This is equivalent to sampling first over $N$ with probability $\sfrac{1}{2}$ as done in \ref{item:mainAlg}, and then sampling another $\sfrac{1}{e}$ fraction over the remaining elements as done in \ref{item:clasicalSec2}. 
	Let $S'$ be a random subset of $S$ obtained by subsampling each element of $S$ with probability $\sfrac{1}{(e+1)}$ independently.
	Note that $S'$ is then a random set containing each element of $N$ with probability $\sfrac{1}{2e}$ independently.
	We then use $S'$ to set a threshold, and select the first element arriving in $N \setminus S$ with weight above the threshold (if any). 
	This procedure is equivalent to running branch \ref{item:clasicalSec2} in the original algorithm.
	
	\item If branch \ref{item:osp1} is chosen, the algorithm sets $p=\sfrac{1}{2}$ to obtain a sample set $S$, and then runs $\mathrm{OSP}(\mrestrict{\M}{N \setminus S}, 1)$ on the remaining non-sampled elements. 
	This does not have any impact on the uniform assignment of the weights $w$ to the elements, since it only depends on the sampled elements $S$ but not the weights of its elements. 
	Thus, the RA-MSP subinstance given by $\mrestrict{\M}{N\setminus S}$ on which we use OSP indeed assigns weights of $w$ uniformly at random to elements, as required.
	Since the guarantees of OSP hold under adversarial arrival order, this is equivalent to running branch \ref{item:osp1} in the original algorithm.
	
	\item Finally, if branch \ref{item:osp2} is chosen, the algorithm sets $p=\sfrac{3}{4}$ to obtain a sample set $\overline{S}$. 
	Let $S'$ be a random subset of $\overline{S}$ obtained by subsampling each element of $\overline{S}$ with probability $\sfrac{2}{3}$ independently.
	Note that $S'$ is then a random set containing each element of $N$ with probability $\sfrac{1}{2}$ independently, while $\tilde{S} \coloneqq \overline{S} \setminus S'$ is random set containing each element of $N$ with probability $\sfrac{1}{4}$.
	We then use $S'$ to simulate the sample set $S$ used in branch \ref{item:mainAlg} of the original algorithm, and $\tilde{S}$ to simulate the sample set $S'$ used in branch \ref{item:osp2} of the original algorithm.
	As discussed in the case above, this construction does not have any impact on the uniform assignment of the weights $w$ to the elements. 
	Thus, the RA-MSP subinstance given by $\mrestrict{\M}{N\setminus \overline{S}}$ on which we use OSP indeed assigns weights of $w$ uniformly at random to elements, as required.
	Finally, using that the guarantees of OSP hold under adversarial arrival order, it follows that this procedure is equivalent to running branch \ref{item:osp2} in the original algorithm.
\end{itemize}
 
\printbibliography

\appendix
\section{Omitted Proofs from \cref{sec:curveEstimate}}\label{sec:cuveEstimate-proofs}

The following result is used to show property (ii) in the proof of \eqref{eq:concentrationCard} in \cref{thm:concentrationBounds}.

\begin{lemma}\label{lem:density-union}
	Let $\M = (N, \I)$ be a matroid, let $\M_{h} = (N, \I_{h})$ denote its $h$-fold union for some $h \in \Z_{\geq 1}$, and let $r_{h}$ denote the rank function of $\M_{h}$.
    Let $Q \subseteq N$ and $e \in N \setminus Q$ be such that $r_{h}(Q \cup \{e\}) = r_{h}(Q)$. Then $e \in \mspan(D(Q, h))$, where $D(Q, h)$ is as defined in \cref{eq:DSLambda}.
\end{lemma}
\begin{proof}
    Let $r$ denote the rank function of $\M$.
    By the matroid partitioning theorem of Nash-Williams, we have
    \begin{equation}\label{eq:nash-williams}
        r_h(U) = \min_{B\subseteq U}\left\{
        |U\setminus B| + h r(B)
        \right\} \qquad \forall U\subseteq N.
    \end{equation}
    Let $B^{*} \subseteq Q$ be such that $r_{h}(Q \cup \{e\}) = \card{(Q \cup \{e\}) \setminus B^{*}} + h r(B^{*})$. Note that $e \in B^{*}$, as otherwise
    \[
        r_{h}(Q)
            \leq \card{Q \setminus B^{*}} + h r(B^{*})
             =  \card{(Q \cup \{e\}) \setminus B^{*}} - 1 + h r(B^{*})
             =  r_{h}(Q \cup \{e\}) - 1, 
    \]
    where the first inequality follows from~\eqref{eq:nash-williams}; however, this contradicts $r_{h}(Q \cup \{e\}) = r_{h}(Q)$. Additionally, note that $r(B^{*} \setminus \{e\}) = r(B^{*})$, as otherwise
    \begin{align*}
        r_{h}(Q)
            &\leq \card{Q \setminus (B^{*} \setminus \{e\})} + h r(B^{*} \setminus \{e\})
                =  \card{(Q \cup \{e\}) \setminus B^{*}} + 1 + h (r(B^{*}) - 1)\\
            &=  r_{h}(Q \cup \{e\}) - h + 1, 
    \end{align*}
    where the first inequality uses again~\eqref{eq:nash-williams}; this again contradicts $r_{h}(Q \cup \{e\}) = r_{h}(Q)$. Moreover, by putting $r_{h}(Q \cup \{e\}) = r_{h}(Q)$ together with $e \in B^{*}$ and $r(B^{*} \setminus \{e\}) = r(B^{*})$, we get:
    \[
        r_{h}(Q)
            = r_{h}(Q \cup \{e\})
            = \card{(Q \cup \{e\}) \setminus B^{*}} + h r(B^{*})
            = \card{Q \setminus (B^{*} \setminus \{e\})} + h r(B^{*} \setminus \{e\}).
    \]
    This implies, due to~\eqref{eq:nash-williams}, that the set $A^{*} \coloneqq B^{*} \setminus \{e\}$ is a minimizer of $|Q\setminus A| + h r(A) = |Q| - (|A| - h r(A))$ over all $A\subseteq Q$.
    Equivalently, $A^*$ is therefore a maximizer of $|A| - h r(A)$ over all $A\subseteq Q$; hence, $A^* \subseteq D(Q,h)$, because $D(Q,h)$ is the unique maximal maximizer of the same expression.
    Recalling that $e \in \mspan(A^{*})$, we thus get $e \in \mspan(D(Q, h))$, as desired.
\end{proof}
  
\end{document}